\documentclass[letterpaper,11pt]{article}
\usepackage{times}
\newcommand{\blind}{1}
\usepackage{amssymb,amsbsy,amsfonts,amsmath,xspace,amsthm}
\usepackage{mathrsfs}
\usepackage{graphicx}
\usepackage{caption}
\usepackage{setspace}
\usepackage{comment}
\usepackage[margin=1in]{geometry}
\usepackage{enumitem}
\usepackage{subcaption}

\newcommand{\p}{\mathbb{P}}

\newcommand{\cov}{\mathrm{Cov}}

\newcommand{\argmin}{\mathrm{argmin}}
\newcommand{\tr}{\mathrm{tr}}

\newcommand{\mbzero}{{\mathbf{0}}}

\newcommand{\gcal}{\mathcal{G}}

\newcommand{\bR}{\mathbb{R}}

\newcommand{\bS}{\mathbb{S}}

\newcommand{\norm}[1]{\Vert{#1}\Vert}

\newcommand{\BEAS}{\begin{eqnarray*}}
\newcommand{\EEAS}{\end{eqnarray*}}
\newcommand{\BEA}{\begin{eqnarray}}
\newcommand{\EEA}{\end{eqnarray}}
\newcommand{\BEQ}{\begin{equation}}
\newcommand{\EEQ}{\end{equation}}
\newcommand{\BIT}{\begin{itemize}}
\newcommand{\EIT}{\end{itemize}}

\newcommand{\Rank}{\mathop{rank}}

\newcommand{\diag}{\mathop{diag}}

\newtheorem{thm}{Theorem}

\newtheorem{defi}{Definition}
\newtheorem{algo}{Algorithm}

\newtheorem{lem}{Lemma}

\newtheorem{ass}{Assumption}
\theoremstyle{definition}

\usepackage{multirow}
\usepackage{epstopdf}
\usepackage{epsfig}
\usepackage[colorlinks,citecolor=blue]{hyperref}

\usepackage{url}
\usepackage[toc,page]{appendix}
\usepackage{float}
\usepackage{natbib}
\usepackage{color}
\usepackage{verbatim}

\makeatletter
\newcommand{\printfnsymbol}[1]{%
  \textsuperscript{\@fnsymbol{#1}}%
}
\makeatother

\def\text#1{\mbox{\rm #1}}

\begin{document}




\if1\blind
{
  \title{\bf Compressed spectral screening for large-scale differential correlation analysis with application in selecting Glioblastoma gene modules }
\author{Tianxi Li, Xiwei Tang\\ Department of Statistics, University of Virginia
\and Ajay Chatrath\\ Department of Biochemistry and Molecular Genetics, University of Virginia}
\date{\today}

  \maketitle
} \fi

\if0\blind
{
  \bigskip
  \bigskip
  \bigskip
  \begin{center}
    {\LARGE\bf Compressed spectral screening for large-scale differential correlation analysis with application in selecting Glioblastoma gene modules}
\end{center}
  \medskip
} \fi

\bigskip

\abstract{
Differential co-expression analysis has been widely applied by scientists in understanding the biological mechanisms of diseases. However, the unknown differential patterns are often complicated; thus, models based on simplified parametric assumptions can be ineffective in identifying the differences. Meanwhile, the gene expression data involved in such analysis are in extremely high dimensions by nature, whose correlation matrices may not even be computable. Such a large scale seriously limits the application of most well-studied statistical methods. This paper introduces a simple yet powerful approach to the differential correlation analysis problem called compressed spectral screening. By leveraging spectral structures and random sampling techniques, our approach could achieve a highly accurate screening of features with complicated differential patterns while maintaining the scalability to analyze correlation matrices of $10^4$--$10^5$ variables within a few minutes on a standard personal computer. We have applied this screening approach in comparing a TCGA data set about Glioblastoma with normal subjects. Our analysis successfully identifies multiple functional modules of genes that exhibit different co-expression patterns. The findings reveal new insights about Glioblastoma's evolving mechanism. The validity of our approach is also justified by a theoretical analysis, showing that the compressed spectral analysis can achieve variable screening consistency.}

\section{Introduction}

\subsection{Differential expression and co-expression analysis} \label{sec:intro_dca}
High-throughput RNA sequencing (RNA-seq) data have recently drawn great attention in genomic studies \citep{anders2010differential, soneson2013comparison, zhang2014comparative}.  As a powerful tool to quantify the abundance of mRNA transcripts in a sample, RNA-seq data have increasingly been used to identify differentially expressed genes associated with specific biological and clinical phenotypic variations \citep{krupp2012rna,lonsdale2013genotype}.  
For example, differential gene expression analysis can be adopted to detect mRNA transcripts with varying abundances in tumor samples versus normal tissue samples  \citep{wan2015bioxpress,li2016cancer}.  RNA-seq data therefore represent a popular alternative to microarrays in such work \citep{tarazona2011differential, costa2017rna}.

Conventional differential expression analysis focuses on comparing marginal gene expression levels between conditions or experimental groups \citep{wang2007detecting, li2011adaptively, trapnell2012differential, sun2015isodot, zhao2017sparse, dadaneh2018bnp}.  A complete understanding of the molecular basis of phenotypic variation also requires characterizing the interactions between genetic components \citep{ballouz2015guidance, van2018single}.  Many clustering algorithms have been proposed to identify groups of co-expressed genes (see \cite{sarmah2021study} for a review of popular tools used to analyze RNA-seq data). In many cases, researchers’ primary interest lies in discerning co-expression pattern variation among genes.  For example,   \citep{hudson2009differential} discovered that the key myostatin gene is not differentially expressed between two cattle breeds but exhibits distinct co-expression structures of the involved causal mutations. Statistical methods appropriate for this type of analysis remain understudied.

A widely used strategy for statistical differential co-expression analysis is modeling differential networks (see \cite{shojaie2021differential} for a detailed review). In particular, many approaches involve fitting Gaussian graphical models with certain sparsity assumptions  \citep{friedman2008sparse, yuan2010high, cai2011constrained}.  One method entails jointly estimating multiple precision matrices  \citep{chiquet2011inferring, guo2011joint,saegusa2016joint}, wherein the individual matrices are often assumed to be sparse.   \cite{danaher2014joint}  introduced the fused graphical lasso approach to encourage fusion over entries of estimated precision matrices.  \cite{zhao2014direct}  extended \cite{cai2011constrained}'s constraint $\ell_1$-minimization algorithm to directly estimate the differential network, that is, the difference of two precision matrices. \cite{xia2015testing} proposed a testing framework to identify the difference between two partial dependence networks. \cite{yuan2017differential} proposed a D-trace loss function with Lasso penalty to estimate a differential network with sparsity. Despite its relative popularity, this class of approaches suffers from many drawbacks in our application scenario. First, the partial correlation structure is usually estimable based on the Gaussian assumption of data. However, non-normal distributions of gene expression data have been widely documented  \citep{marko2012non}, and growing evidence suggests a non-linear relationship between genes \citep{yang2021model}. Second, partial dependence interpretation is less straightforward than marginal dependence and thus appears less frequently in biological and medical studies. Third, due to the noisy nature of data, medical researchers typically prefer to adopt robust versions of dependence metrics \citep{chatrath2020pan}, such as the Spearman correlation coefficient; less is known about statistical methods related to robust partial dependence. Finally, differential analysis based on partial dependence tends to be computationally expensive and unable to tackle the scale of our problem. We therefore focus on the differential analysis of marginal co-expression dependence in this paper. 
 
 The differential analysis of marginal correlations has been considered in several contexts, mainly when testing variations. \cite{schott2007test} and \cite{li2012two} consider testing two covariance matrices that are different in most entries, while \cite{cai2013two,cai2016inference}, and \cite{chang2017comparing} introduced methods appropriate for cases when differential patterns are sparsely distributed. In our problem (to be introduced in the next section), we aim to discover differentially co-expressed genes between cancer patients and healthy subjects. This differential pattern does not exist everywhere in a correlation matrix. Yet when considering the correlation structure, each gene simultaneously affects all its correlation entries. The expected differential pattern of the correlation matrices should hence be globally sparse but locally dense. Therefore, neither of the above two scenarios reasonably approximates our case. \cite{zhu2017testing} rrevealed that spectral properties are far more effective for detecting this type of differential signal.  In particular, \cite{zhu2017testing} proposed using the spectral norm of differential covariance matrices as the test statistic for the differential hypothesis. However, differential gene selection is not studied in their work. 
 
In this paper, we consider a general correlation matrix which enables us to adopt more flexible data-based correlation measurements and thus better capture the co-expression structure. Our proposed method is grounded in spectral methods. Compared with prior work on the differential analysis of covariance/correlation matrices, our primary contributions are two-fold.  Firstly, we introduce a “spectral screening” method, which identifies the differentially co-expressed genes instead of simply testing whether there are differences. Our method considers block differences, which are more informative for our application than the aforementioned studies. Our approach can also be seen as a submatrix localization method. However, different from existing matrix localization methods, we do not enforce additional assumptions on the differential pattern  (e.g., the sign constraint in \cite{butucea2015sharp}, or the constant difference in \cite{chen2016statistical, cai2017computational,liu2019multiscale}). This generality affords us appreciable advantages in analyzing real-world datasets. Second, we present a simple yet powerful randomized version of the localization method that significantly improves the scalability of our spectral screening approach to handle large-scale differential analysis; the previous mentioned methods cannot do so. Specifically, we calculate a small proportion of correlation differences and then use the spectral properties of the incomplete differential matrix to localize differential variables. This strategy possesses much lower computational complexity while maintaining sound accuracy and theoretical guarantees.
 
\subsection{Motivating application and the data set} \label{sec:dataset}

The Cancer Genome Atlas (TCGA) is a large set of tumor and normal tissue samples collected from over 10,000 patients cataloging molecular abnormalities observed across 33 cancer types  \citep{cancer2015comprehensive}. The Genotype Tissue Expression (GTEx) project is a large-scale sequencing project including more than 9,000 samples of 53 different tissues from over 500 healthy individuals  \citep{lonsdale2013genotype,gtex2015genotype,wang2018unifying}. Scholars have identified differentially expressed genes in multiple tumor types via differential expression analysis by comparing tumor samples from TCGA to normal tissue samples from GTEx after batch-effects correction. However, gliomas were excluded from these studies due to a lack of normal brain samples in TCGA  \citep{wang2018unifying}. Additionally, it is impossible to normalize the samples for differences in batch effects across studies when using standard approaches given this absence of normal brain samples. This paper focuses on differential patterns of correlation rather than expression levels. As such, our approach offers an alternative to conventional differential analysis to discern meaningful differential patterns when normalization is not possible through standard methods.

Gliomas are highly invasive primary brain tumors that are challenging to resect neurosurgically without substantial patient morbidity. Approximately 20,000 patients are diagnosed with gliomas each year in the United States. Gliomas can be divided into several grades according to severity. Grade II and Grade III gliomas are lower-grade gliomas; Grade IV gliomas are also known as glioblastoma multiforme. Glioblastoma multiforme is highly lethal and has a five-year survival rate of less than 5\%. Although lower-grade gliomas do not progress as quickly as glioblastoma multiforme, they are essentially uniformly lethal and can progress to glioblastoma multiforme; the median period of survival for lower-grade gliomas is seven years. Typically, lower-grade gliomas and glioblastoma multiforme are first managed surgically, followed by chemotherapy and radiation therapy   \citep{cancer2015comprehensive, bauman2009adult,yan2012molecular,louis20162016}. The 2016 World Health Organization’s guidelines for brain tumors suggest that gliomas should be classified based on molecular characteristics \citep{louis20162016}. Given recent recognition of the roles of molecular characteristics in classifying gliomas and the lethality of these tumors, it is crucial to identify genes of interest warranting further study. We compare glioblastoma multiforme samples from the TCGA dataset to normal brain samples from the GTEx dataset using the approach described in this paper. As mentioned, lower-grade gliomas (Grade II or III) have a better prognosis than glioblastoma multiforme (Grade IV) but can progress to glioblastoma multiforme. We hypothesize that lower-grade gliomas would show an intermediate degree of dysregulation between glioblastoma multiforme and normal brain tissue. Specifically, we use lower-grade gliomas as a validation set in this study, as suggested in prior research \citep{lonsdale2013genotype, cancer2015comprehensive, gtex2015genotype, yan2012molecular,louis20162016}. The key challenge in this setting lies in the dimensionality of the problem: we have a few hundred observations among 51,448 genes. It would be impossible to calculate all correlation coefficients for this number of genes --- never mind performing in-depth differential analysis. A computationally feasible analysis would remain elusive even with basic filtering. Our method is intended to handle differential correlation analysis of such large-scale problems. A detailed analysis is presented in Section~\ref{sec:data-analysis}.
 
%
%
%
%
%
%
\paragraph{Organization of the paper.} Section \ref{sec:method} introduces the proposed methodology. Practical considerations such as parameter tuning and implementation, and complexity analysis are discussed in Section~\ref{sec:tune}. Section \ref{sec:sim} illustrates the model performance by simulation studies. Section \ref{sec:data-analysis} presents our real data analysis findings. Section \ref{sec:theory} provides theoretical results.  Section \ref{sec:discussion} concludes with further discussions.







\section{Methodology}\label{sec:method}

\paragraph{Notations.} Given a positive integer $p$, define $[p] = \{1, 2, \cdots, p\}$. Let $\bR^{p\times p}$ be the set of all $p\times p$ matrices, and $\bS_+^{p\times p}$ be the set of all $p\times p$ positive definite matrices. Given a square matrix $M$, denote its spectral norm and Frobenius norm by $\norm{M}$ and $\norm{M}_F$, respectively. For two sequences $\{a_n\}$ and $\{b_n\}$, we write $a_n = O(b_n)$ if $a_n/b_n$ is bounded for all sufficiently large $n$. We also write $a_n = o(b_n)$ if $a_n/b_n \to 0$ as $n$ goes to $\infty$, in which case we may also use the notation $a_n \ll b_n.$ Given a $p\times p$ matrix $M$ and two sets $\gcal_1, \gcal_2 \subset [p]$, we write $M_{\gcal_1,\gcal_2}$ as the submatrix from constraining on rows in $\gcal_1$ and columns in $\gcal_2$.

We will describe our algorithms assuming the statistic of interest is the covariance matrix. The study of correlation matrix just needs a rescaling step, so will not be distinguished conceptually.  Assume we have a size-$n_1$ random sample $x_i, i=1, \cdots, n_1$ from a distribution with mean $\mu_1$ and covariance $\Sigma_1$ and a size-$n_2$ sample  $y_i, i=1, \cdots, n_2$ from a distribution with mean $\mu_2$ and covariance  $\Sigma_2.$
Here, $\mu_1, \mu_2 \in \bR^p$ and $\Sigma_1, \Sigma_2 \in S_+^{p\times p}$. We do not assume special structures (e.g., sparsity) for the individual matrix $\Sigma_1$ (or $\Sigma_2$). We are interested in the situation where $\max(n_1, n_2) \ll p$. Even loading the matrices $\Sigma_1$ and $\Sigma_2$ into a computer’s memory can be challenging, never mind the associated calculation.    Fortunately,  in many differential correlation analyses, it is reasonable to assume that there is only a small set of coordinates $\gcal \subset [p]$ lead to different correlations, with $|\gcal| = m \ll p$. That is,  
$$\Sigma_{1, ij} \ne \Sigma_{2,ij}, \text{~~only if~~} i, j \in \gcal.$$

Our main objective in this paper is to identify $\gcal$. The challenge lies in the problem scale. It is infeasible to examine $\gcal$ even by calculating all pairwise correlations.  Our example in Section~\ref{sec:dataset} involves more than $50,000$ genes , yielding about $2.5 \times 10^9$ marginal correlations for two samples.  Therefore, we need a rapid screening method  that can reduce the problem to a more manageable size for downstream analysis with precise detection.

Specifically, let $\hat{\Sigma}_1$ and $\hat{\Sigma}_2$ be the sample covariance matrices of $X_1$ and $X_2$, respectively. Define $D = \Sigma_1 -\Sigma_2$ and $\hat{D} = \hat{\Sigma}_1 - \hat{\Sigma}_2$. We will first introduce a simple screening method for $\gcal$ based on $\hat{D}$ in Section~\ref{sec:screening}. Then  in Section~\ref{sec:matrix-completion}, we will approximate $\hat{D}$ without computing the full matrix for the screening purposes. 

\subsection{Spectral screening}\label{sec:screening}

Spectral algorithms constitute a family of methods intended to handle structures for large-scale datasets, ranging from classical clustering analysis \citep{shi2000normalized,ng2002spectral}, to more complicated data structures such as text data \citep{anandkumar2015spectral}, time series \citep{hsu2012spectral} and network data \citep{rohe2011spectral,lei2014consistency,li2020community,miao2021informative}. From a macro perspective, our approach echoes network method of \cite{miao2021informative}. The primary difference is that we are working with differential correlation matrices instead of a single network. More importantly, in our case, the dataset has a computationally prohibitive scale; this analysis thus calls for a different spectral method design. We will begin with a discussion of our problem’s spectral properties.

Denote the rank of $D$ by $K$. We have $K \le m$ in the current context. Let $D = U\Lambda U^T$ be the eigen-decomposition of $D$, where $\Lambda = \diag(\lambda_1, \cdots, \lambda_K)$ is a square diagonal matrix of all the nonzero eigenvalues of $D$ (with non-increasing magnitude), and $U = (u_1, \cdots, u_K)$ consists of the corresponding eigenvectors.   Without loss of generality, throughout this paper, we  assume $\gcal$ is the set of the first $m$ variables. That is, $\gcal = [m] $. In this case, let $U_1$ be the matrix of first $m$ rows in $U$ and $U_2$ be the matrix of the last $p-m$ rows.
We have
$$D = \begin{pmatrix}
 D_{\gcal,\gcal} & \mbzero_{m\times (p-m)} \\
\mbzero_{(p-m)\times m}  & \mbzero_{(p-m)\times (p-m)}  \\
 \end{pmatrix}
 =  \begin{pmatrix}
 U_{1}\\
U_{2}\\
 \end{pmatrix} \Lambda U^T
 = \begin{pmatrix}
 U_{1}\Lambda U^T\\
U_{2}\Lambda U^T\\
 \end{pmatrix}.
$$
Such a pattern indicates
$$\mbzero_{(p-m)\times p} = U_{2}\Lambda U^T$$
and leads to
$$U_{2} = \mbzero_{(p-m)\times n}U\Lambda^{-1} = \mbzero_{(p-m)\times K}.$$
This relation suggests a simple strategy to identify $\gcal$ based on the rows in $U$, as these rows fully capture the sparsity pattern of the differential correlation structure. More notably,  the current property does not depend on any additional structural assumption about $D_{\gcal,\gcal}$, and can hence potentially be more general than many other methods. The algorithm based on the above idea is summarized as follows:

\begin{algo}[Spectral screening \texttt{SpScreen}$(\hat{D},K)$]\label{algo:spscreen}
Given the differential correlation matrix $\hat{D}$, a positive integer $K$:
\begin{enumerate}
\item Calculate the eigen-decomposition of $\hat{D} = \hat{U}\hat{\Lambda} \hat{U}^T$ up to rank $K$.
\item Take the screening score $s_i =\norm{\hat{R}_{i \cdot}}$ where $\hat{R} = \hat{U}\hat{\Lambda}^{1/2}$. 
\item Return the spectral score $\{s_i\}$.
\item (Optional) If a threshold vector $\Delta \in \bR_+^p$ is available, select all variables with $s_i > \Delta_i$. 
\end{enumerate}
\end{algo}
\noindent
We expect that $K$ is an integer that is approximately the rank of $D$. It is usually unknown and will be treated as a tuning parameter. The strategy of tuning $K$ will be discussed in Section~\ref{sec:tune}.

\subsection{Fast approximation by random sampling}\label{sec:matrix-completion}
Algorithm~\ref{algo:spscreen} requires the input of $\hat{D} = \hat{\Sigma}_1 - \hat{\Sigma}_2$. In the scenario of our application problems, $p$ is large, and both $\hat{\Sigma}_1$ and $ \hat{\Sigma}_2$ tend to be dense; it would therefore be impractical to calculate the matrices. Alternatively, we resort to approximation methods for screening. Our solution is based on the following two observations: 1) Algorithm~\ref{algo:spscreen} only needs  the leading eigenspaces; 2) the rank of $\hat{D}$ is at most $n_1+n_2$, much smaller than $p$. Because low-rank matrices can be represented efficiently with far fewer entries than $O(p^2)$ \citep{chatterjee2015matrix,li2016network,abbe2017entrywise,chen2019noisy}, it becomes possible to approximate $\hat{D}$ or its eigen-space without needing to know all the entries.  Specifically, we sample a subset of entries in $D$, with each entry being sampled independently with a pre-specified probability $\rho$. Instead of calculating the full $\hat{D}$ for all $p(p-1)/2$ entries, we calculate the sample covariance values on these sampled positions only. Depending on the problem size, using a small $\rho$ (e.g., $\rho < 0.05$) can greatly conserve computational time and memory. Let $\tilde{D}$ be the incomplete matrix with the only values on the subsampled entries. We then approximate the spectral structure of $D$ based on the sparse approximation $\tilde{D}$ with only those sampled entries filled. Analogous to matrix completion problems, we impute the missing entries by zeros and then calculate the corresponding eigenstructures as our approximation. The independent sampling strategy allows for accurate eigenstructures. Our two-step screening algorithm is summarized below. In light of its connections with compressed sensing and matrix completion problems,  we call it \underline{C}ompressed \underline{S}pectral \underline{S}creening (CSS).

\begin{algo}[Compressed spectral screening: \texttt{CSS}$(X_1,X_2, \rho,K)$]\label{algo:css}
Given observation matrices $X_1, X_2$, sampling proportion $\rho$ and $K$. Initiate zero (sparse) matrices $\tilde{D}, \Omega \in \bR^{p\times p}$. Do the following:
\begin{enumerate}
\item For each $(i,j) \in [p]\times [p]$ with $i < j$
\begin{enumerate}
\item Sample $\Omega_{ij} \sim \text{Bernoulli}(\rho)$;
\item If $\Omega_{ij} = 1$, calculate $\hat{D}_{ij} = \cov(X_{1,i\cdot},X_{1,j\cdot}) -\cov(X_{2,i\cdot},X_{2,j\cdot})  $. 
\end{enumerate}
\item Set $\tilde{D}_{ij}  = \hat{D}_{ij}/\rho$.
\item Return \texttt{SpScreen}$(\tilde{D},K)$.
\end{enumerate}
\end{algo}

This subsampling strategy can greatly accelerate spectral screening for large-scale datasets, assuming well-designed implementation. The subsampling and covariance calculation are each simple to implement in a parallel setting; as such, the computational speed can be further improved when a distributed system is available. The CSS algorithm also removes the crucial memory constraint of the originally infeasible problem and renders the computation more scalable. 

\section{Practical considerations: tuning and computation}\label{sec:tune}
\subsection{Tuning parameter selection}
\paragraph{Tuning the rank $K$.}  $K$ be tuned by cross-validation as in many matrix completion problems \citep{chi2019matrix}. However, given the scope of our problem, a full-scale cross-validation is infeasible. We therefore replace the $K$-fold cross-validation procedure with basic random sampling validation, coupled with Algorithm~\ref{algo:css}. Each time, in the random sampling step of the CSS algorithm, we sample additional pairs of variables and validate prediction performance using different tuning parameters by running the eigendecomposition only once. This procedure is shown in Algorithm~\ref{algo:css-tune}.

\begin{algo}[Compressed spectral screening with tuning]\label{algo:css-tune}
Given observation matrices $X_1, X_2$, sampling proportion $\rho$, validation proportion $\tau$, and $K_{l} < K_{u}$. Initiate zero sparse matrices $\hat{D}, \tilde{D}, \Omega, \tilde{\Omega}, \Phi \in \bR^{p\times p}$. Do the following:
\begin{enumerate}
\item For each $(i,j) \in [p]\times [p]$ with $i < j$
\begin{enumerate}
\item Sample $\tilde{\Omega}_{ij} \sim \text{Bernoulli}\left((1+\tau)\rho\right)$; sample $\Phi_{ij} \sim \text{Bernoulli}(\frac{1}{1+\tau})$. Let $\Omega_{ij} = \tilde{\Omega}_{ij}\cdot \Phi_{ij}$. Notice that marginally, $\Omega_{ij}$ can be see as random Bernoulli samples with probability $\rho$.
\item If $\tilde{\Omega}_{ij} = 1$, calculate $\hat{D}_{ij} = \cov(X_{1,i\cdot},X_{1,j\cdot}) -\cov(X_{2,i\cdot},X_{2,j\cdot}) $. If $\Omega_{ij} = 1$, set $\tilde{D}_{ij} = \hat{D}_{ij}/\rho.$ 
\end{enumerate}
\item Calculate the partial eigen-decomposition up to rank $K_u$ for $\tilde{D}$, denoted by $\tilde{D} = \tilde{U}\tilde{\Lambda}\tilde{U}^T$.
\item For $K_l \le K \le K_u$:
\begin{enumerate}
\item Use the rank $K$ approximation  to predict entries for positions $(i,j)$ such that $\tilde{\Omega}_{ij}=1,{\Omega}_{ij} = 0$, calculated by $\hat{D}_{ij}^{(K)} = \tilde{U}_{i,1:K}\tilde{\Lambda}_{1:K,1:K}\tilde{U}_{j,1:K}^T$. For $K > K_l$, this can be computed by 
$$\hat{D}_{ij}^{(K)} = \hat{D}_{ij}^{(K-1)} + \tilde{\lambda}_K\tilde{U}_{iK}\tilde{U}_{jK}.$$
\item Calculate the loss $L_K = \sum_{(i,j):\tilde{\Omega}_{ij}=1,{\Omega}_{ij} = 0}(\hat{D}_{ij}-\hat{D}^{(K)}_{ij})^2$.
\end{enumerate}
\item Select $\hat{K} = \argmin_{K_l \le K \le K_u}L_K.$
\item Return $\texttt{SpScreen}(\tilde{D},\hat{K})$. Note that this step can use the previous partial decomposition in Step 2.
\end{enumerate}
\end{algo}
In Algorithm~\ref{algo:css-tune}, the estimated $\hat{D}^{(K)}$ is only based on the pairs with $\Omega_{ij}=1, \tilde{\Omega}_{ij}=1$, and the pairs with $\Omega_{ij}=0, \tilde{\Omega}_{ij}=1$ constitute the validation set. By default, we always use $\tau = 0.1$ in all experiments. It is easy to see that in the current context of two data matrices, the natural range would be $ K_l = 2$. Also, since we observe only $\rho$ proportion of the entries from $D$, which is a most rank $n_1+n_2$. To ensure a reasonable signal to noise ratio, we can constraint $K_u = \rho(n_1+n_2)$. Empirically, theses choices give very effective range for performance in all of our evaluations. Additional constraints from side information can be further enforced to reduce the computational cost and model selection variance.

\paragraph{Selection of sampling proportion $\rho$. } The larger $\rho$ is, the more information we can include for screening. We would always prefer to use a larger $\rho$ within the affordable computational limit. Also, it can be expected that if $\rho$ is too small, eventually we will not have sufficient information for meaningful screening.  Here we introduce an ad hoc rule to choose the lower bound of $\rho$ based on our theoretical analysis (see Section~\ref{sec:theory}), which works well in all of our experiments.  Matrix $\hat{D}$ has $p(p+1)/2$ entries calculated from $(n_1+n_2)p$ entries from the raw data. Therefore, intuitively, we need more than $(n_1+n_2)p$ entries, so as to retain reasonable amount of information to make the recovery problem feasible; accordingly, $\rho > 2(n_1+n_2)/(p+1)$.

\paragraph{Determining the selection threshold.} In the previous algorithms, e.g., Algorithm~\ref{algo:spscreen}, we did not specify how to determine the cutoff of the spectral scores to identify $\gcal$; this decision may be case-specific. In many situations, data analysts already have a sense of how many variables will be selected. Similarly, when the CSS is used for initial screening to reduce the problem size for a refined analysis (as in our analysis in Section~\ref{sec:data-analysis}), the threshold can be set to produce a feasible number of variables for the downstream algorithm. In other scenarios, we need an automatic and data-driven strategy to determine differential variables based on spectral scores. Resampling methods have been widely applied to determine reasonable selection in many problems \citep{meinshausen2010stability,Lei2020,le2020}. We assume a similar angle and propose a bootstrap procedure \citep{efron1979bootstrap} to determine $\Delta$. Because the spectral scores of differential variables are large, we should determine our selection by measuring the upper bound of the non-differential variables’ scores that one can expect to have. Although these scores are unknown, we can use bootstrapping to create a null distribution where all variables follow the same covariance structure. The procedure is shown below:
\begin{algo}[Stability selection for spectral screening]\label{algo:boot}
Given data $X_1 \in \bR^{n_1\times p}$, $X_2 \in \bR^{n_2\times p}$, the number of bootstrapping replications, $B$, and $K$ and $\rho$ in the spectral screening algorithm.
\begin{enumerate}
\item Let $S =$ \texttt{SpScreen}$(\tilde{D},K)$.
\item For $b = 1, 2, \cdots, B$, generate the bootstrap samples:
\begin{enumerate}
\item Sample $n_1$ rows from $X_2$ with replacement to stack into a matrix $\tilde{X}_1^{(b)}$, and $n_2$ rows from $X_2$ with replacement to stack into a matrix $\tilde{X}_2^{(b)}$.
\item Calculate the spectral score $S^{(b)}$ of the spectral screening on $\tilde{X}_1^{(b)}$ and $\tilde{X}_2^{(b)}$.
\end{enumerate}
\item Select each $j$ with $\frac{\sum_b I(S_j > S_j^{(b)})}{B} \ge 0.99$.
\end{enumerate}
\end{algo}
It is worth noting that here, the thresholds for the $p$ variable generally differ because the bootstrap procedure adapts to the variability of each variable. This method is quite effective in our evaluation  (see Section~\ref{sec:sim}). The stability selection requires employing spectral screening $B$ times, which carries a high computational cost.  Therefore, \emph{for large-scale differential analysis}, we recommend the following strategy: (1) apply compressed spectral screening to conduct a rough screening and significantly reduce the problem size; (2) use the full version of spectral screening on the reduced dataset; (3) adopt stability selection to determine the selection. We use this strategy for glioblastoma gene analysis in Section~\ref{sec:data-analysis}.

\subsection{Sampling implementation and complexity analysis}
The algorithms introduced thus far require efficient implementation to be scalable in practice. In this section, we provide additional computational details and a corresponding complexity analysis. We do not consider the extra tuning step because, as discussed, the tuning cost is in a lower order of complexity.  For simplicity, we assume that $n_1=n_2=n$ in this section. If not, all of the results in this section still hold after replacing $n$ with $n_1+n_2$.

The sampling step of $\Omega$ naive implementation of sampling the Bernoulli generator takes $O(p^2)$ operations, which is excessive when $p$ is large. Instead, the sampling step should be based on generating geometric random numbers. Given any ordered list of the  indices of paired nodes  and map them to $1, 2, \cdots, {{p}\choose{2}}$. Denote this mapping by $\pi: [p]\times [p] \to [ {{p}\choose{2}}] $. 
Notice that the gap between two consecutively sampled indices by the Bernoulli sampler follows a geometric distribution. Therefore, instead of sampling the status of each pair, we can use the geometric distribution to generate the sequence of the sampled positions.

\begin{algo}
Start with $Y = \emptyset$. While $\max\{ y \in Y\}  < {{p}\choose{2}}$, do the following:
\begin{enumerate}
\item Generate $Z = z$ from geometric distribution with probability mass function $\p(Z = z) = (1-\rho)^{z-1}\rho, z = 1, 2, \cdots$
\item If $Y = \emptyset$, append $z$ to $Y$. Otherwise, append $\max\{y \in Y\} + z$ to $Y$.
\end{enumerate}
Return $Y$ as the sampled indices of the order pairs. So $\pi^{-1}(Y)$ is  the sampled pairs of the original matrix.
\end{algo}
According to \cite{bringmann2013exact}, the complexity of generating a geometric random number is $O\left(\log(1/\rho)\right)$. The generating procedure of $\Omega$ according to the algorithm above is $O\left({p\choose{2}}\rho\log(1/\rho)\right)$. Compared with a naive sampling, we save a factor of $\rho\log(1/\rho)$ in complexity. The sampling step also involves computing the covariance values, which needs $O\left({p\choose{2}}\rho n\right)$ complexity.

The next computational chunk is the rank-$K$ eigen-decomposition. Notice that the matrix $\tilde{D}$ is now a sparse matrix with roughly $\rho p^2$ nonzero entries. The computational complexity of this step depends on the signal distribution of the matrix, but it is generally estimated \citep{saad2011numerical} to have a complexity in the order of $O(K\rho p^2 + K^2p)$. The remaining algorithm steps have a lower-order complexity.

Let $d = p\cdot \rho$ be the expected number of observed entries in each row of $\hat{D}$. In summary, the expected complexity for the whole procedure is
\begin{equation}\label{eq:AverageComplexity}
O\left(\rho p^2\log(1/\rho)+p^2\rho n+K\rho p^2 + K^2p\right) = O\left(p\cdot \max\left\{dK, d\log{p}, dn, K^2 \right\}\right).
\end{equation}

Regarding the order of $\rho$, a natural requirement is that we should at least observe observations for each row (or column) of the full matrix $\hat{D}$; otherwise, there is no chance of recovering corresponding information about the missing variable. Therefore, we cannot use an arbitrarily small $\rho$. It is known \citep{erds1960evolution} that we should ensure $d  > 2\log{p}$ (or equivalently, $\rho > 2\log{p}/p$) to make sure that we have observations for each row (or column) from $\hat{D}$, with high probability.  This is always what we assume. A few other natural constraints also exist: $K$ should always be smaller than $2n$; For a reasonable differential recovery, we consider only the cases when $\log{p} \ll n$. Our recommended setup $\rho \approx n/p$ in the previous section satisfies all the requirements. Combining these constraints together, the complexity \eqref{eq:AverageComplexity} is simplified to
\begin{equation}\label{eq:AverageComplexity-final}
O\left(pn^2\right).
\end{equation}
The whole procedure also stores $O(np)$ numbers in memory, with the help of sparse matrix data structure. In comparison, the naive method requires $O(p^2n)$ computational complexity with $O(p^2)$ for memory.  The CSS procedure thus results in a saving factor of  $n/p$ in both timing and memory with in its sequential implementation, which is huge in high-dimensional settings. The timing advantage is even larger with parallel computing.


\section{Simulation}\label{sec:sim}
In this section, we use simulated data to evaluate the proposed method for differential analysis on both Pearson-type correlation and Spearman-type correlation. We also compare the proposed method with several recent methods for submatrix localization in both small-scale problems ($n_1=n_2=100, p=2000$) and large-scale problems ($n_1=n_2=100, p=40000$). In the first setting, calculating the full correlation matrices is still computationally feasible and all methods of submatrix localization can be used. We specifically include the \emph{spectral projection} method of \cite{cai2017computational}, which is based on a model property studied in \cite{cai2017computational}. We then embed this method into our subsampling strategy to extend its scalability. Moreover, we include the adaptive large average submatrix (LAS) algorithm proposed in  \cite{liu2019multiscale}, as an improved version of \cite{shabalin2009finding}, along with the Golden Section Search algorithm (from the same paper) as two other benchmarks. However, both methods require a full correlation matrix as input and are hence inapplicable to large-scale problems.

The data $x_i$ and $y_i$ ($i = 1, \cdots, 100$) are generated from multivariate normal distributions with zero mean and covariance matrices. The covariance structures are given by the so-called  spiked covariance model \citep{johnstone2009consistency}
\begin{equation}\label{eq:spikedmodel}
\Sigma_1 = I + v_1v_1^T, \text{~~~~and~~~~} \Sigma_2 = I + v_2v_2^T,
\end{equation}
where $v_1$ has its first $50$ entries generated from $N(1,0.2)$ and the rest are zeros; $v_2$ has its $51$st to $100$th entries generated from $N(1,0.2)$ and the rest as zeros. Therefore, the two models have different  correlations only on the first 100 coordinates. The spiked covariance model includes the differential correlation structure with the models in \cite{chen2016statistical,cai2017computational,liu2019multiscale} as special cases, but is more general because the differential components do not have to be a constant.

We will evaluate the screening (or, differentially correlated variable selection) accuracy of the true differential variables by their sensitivity and specificity. Sensitivity refers here to the proportion of differential variables that are retained, while specificity reflects the proportion of null variables that are filtered out. For our method and that of  \cite{cai2017computational}, we can generate the full ROC curve with respect to sensitivity and specificity by varying the number of selected variables which, as mentioned, represents a core advantage in large-scale analysis.  The adaptive LAS and Golden Section Search algorithms automatically produce a final separation of the data and therefore only result in a single point in the ROC plane given each instantiation of the data. The timing of computation is another important aspect we wish to evaluate given our focus on computationally feasible methods for large-scale problems. We implement our approach and the spectral projection algorithm in R. The adaptive LAS and Golden Section Search algorithms are based on the implementation provided in \cite{liu2019multiscale}, available for download at  \texttt{https://github.com/nozoeli/biclustering}. In all the configurations, we repeat the analysis for 50 independent instantiations. 

\subsection{Results for small-scale problems}

For the same problems with $n_1=n_2=100$ and $p=2000$, we test the full data versions of our approach, the spectral projection method, as well as the corresponding subsampled version with $\rho = 0.05, 0.1, 0.2$. We consider both Pearson's correlation and Spearman's correlation as the statistic for analysis. Table~\ref{tab:small-compare-rho} shows the average area under the ROC curve (AUC) of the 50 replications for all three values of $\rho$ and the full version ($\rho=1$). Findings in the table indicate that the methods’ performance is fairly robust to the three sampling proportions in this case. Moreover, no observable differences in performance emerge when using either Pearson’s or Spearman’s correlation as the metric.

\begin{table}[ht]
\centering
\caption{The AUC for three subsampling proportions from our method and the spectral projection of \cite{cai2017computational} in small-scale problems.}\label{tab:small-compare-rho}
\begin{tabular}{l|cccc|cccc}
  \hline
 &  \multicolumn{4}{c|}{Pearson} &  \multicolumn{4}{c}{Spearman}  \\ 
  \hline
$\rho$ & 0.05 & 0.1 & 0.2& 1 (full)&0.05 & 0.1 & 0.2 & 1 (full) \\ 
  \hline
Our method (spectral screening) & 0.955 & 0.996 & 0.991 & 0.998  & 0.951 & 0.995 & 0.990 & 0.998 \\ 
Spectral projection \citep{cai2017computational} & 0.746 & 0.752 & 0.751 & 0.754 & 0.744 & 0.750 & 0.753 & 0.754 \\ 
   \hline
\end{tabular}
\end{table}

\begin{figure}[H]
\centering
\begin{subfigure}[t]{0.48\textwidth}
\centering
\includegraphics[width=\textwidth]{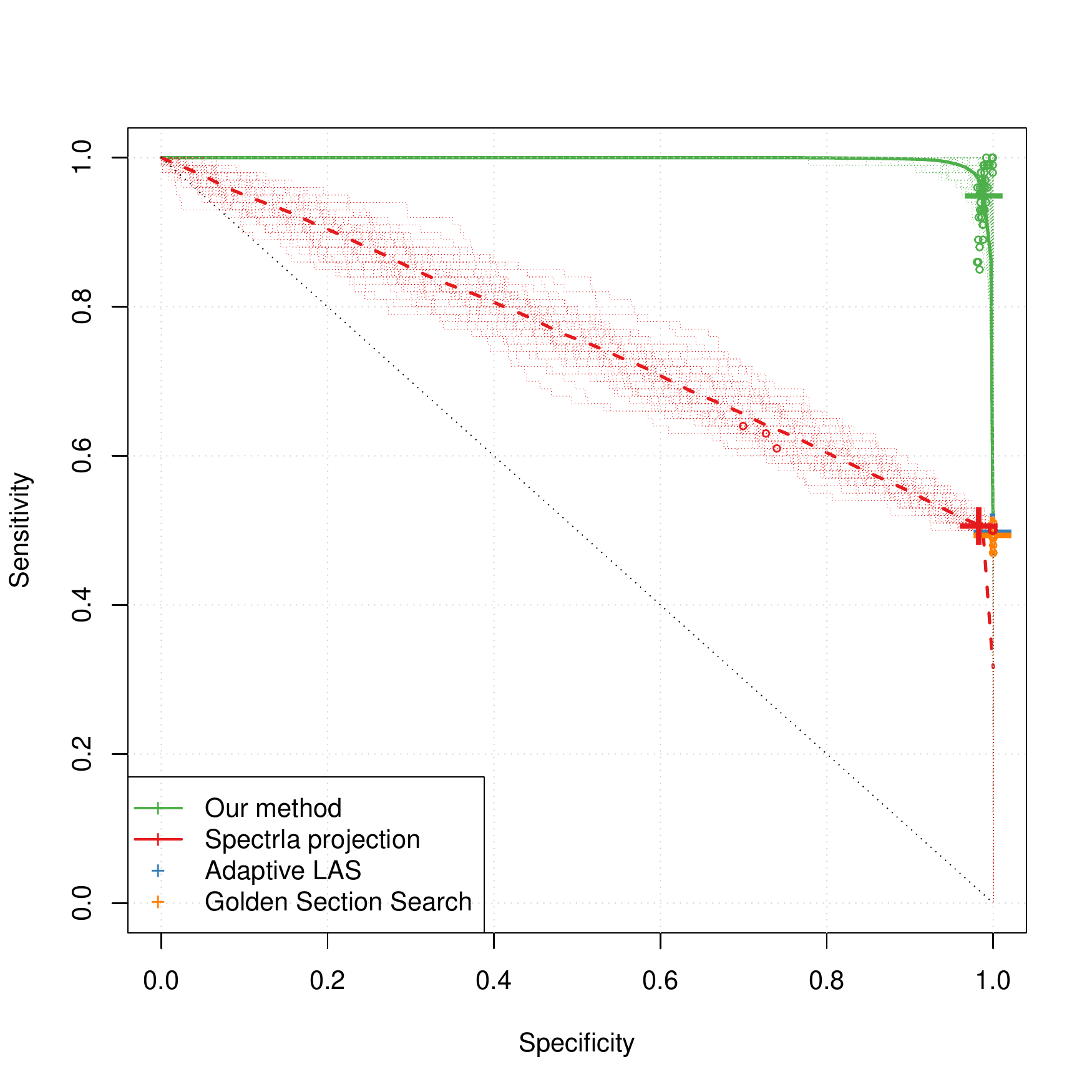}
\caption{Pearson's correlation}
\label{fig:small-Pearson}
\end{subfigure}
\hfill
\begin{subfigure}[t]{0.48\textwidth}
\centering
\includegraphics[width=\textwidth]{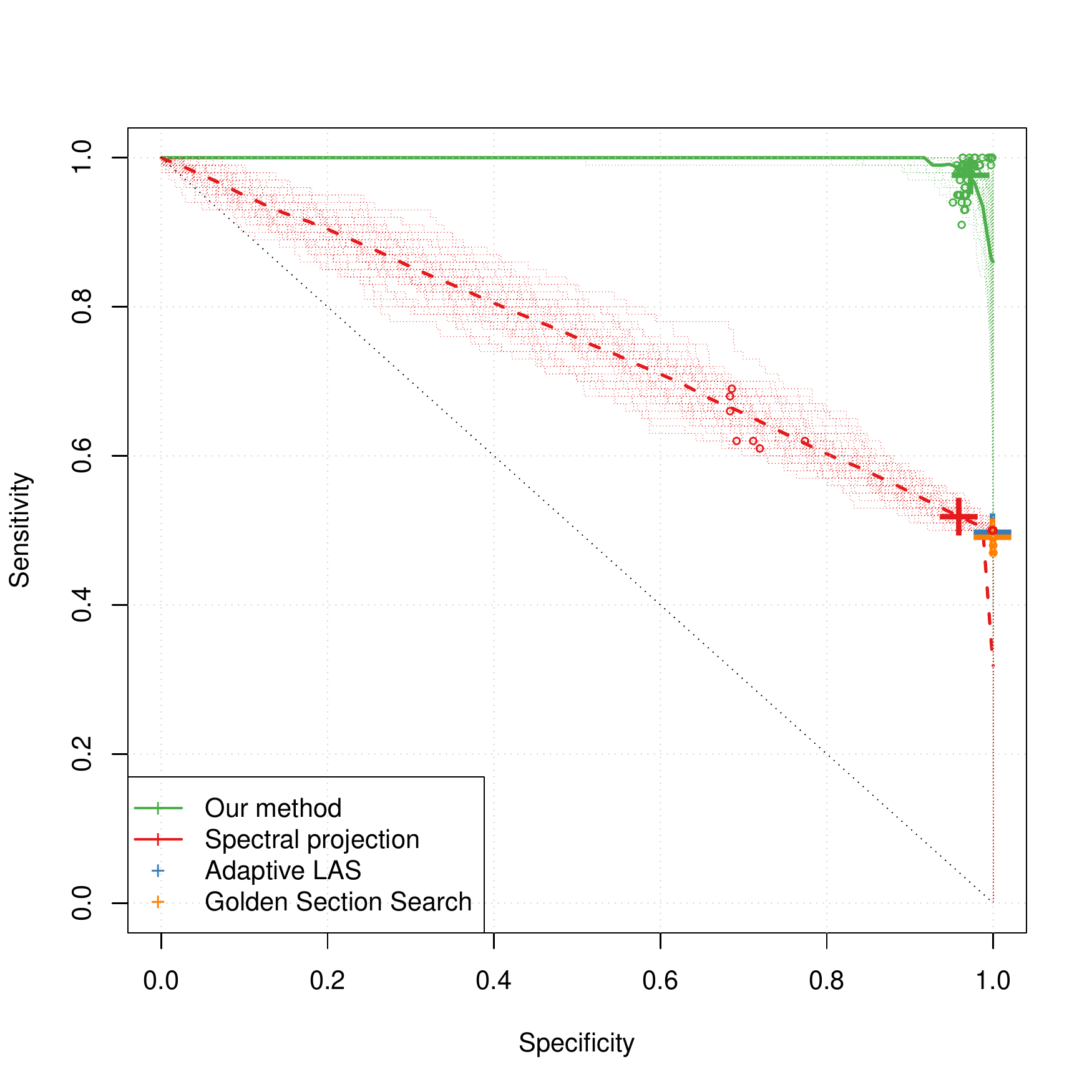}
\caption{Spearman's correlation}
\label{fig:small-Spearman}
\end{subfigure}
\caption{Sensitivity–specificity ROC plot for all four methods with small-scale problems  ($n_1=n_2=100, p=2000$). The ``+" signs are the average point of the single estimates of each method.}
\label{fig:small-ROC}
\end{figure}

Next, we use the variant of $\rho=0.1$ in our approach and the spectral projection method for comparison with the adaptive LAS and Golden Section Search algorithms based on the full correlation matrices. Figure~\ref{fig:small-ROC} depicts the average ROC curves of our method and the spectral projection, with 50 individual ROC curves in the background. The adaptive LAS and Golden Section Search algorithms can produce single selections but not ROC curves. 
Thus, we show 50 individual points on the background with their averaged sensitivity and specificity indicated by ``+". Spectral projection can also produce a single selection, based on the clustering properties indicated in \cite{cai2017computational}. The 50 single instantiations and their average appear in the figure as small background circles and a solid ``+". For our case, we use our bootstrap method to produce a single selection. The 50 individual selection points and the average are shown in the same way.

Both the adaptive LAS and Golden Section Search algorithms offer good specificity, but they tend to be overly conservative and do not allow for the flexibility of choosing a threshold. Spectral projection accommodates subsampling and produces a full ROC curve. Its performance is similar to the other methods in early stages; however, it suffers from poor tradeoffs between sensitivity and specificity if one wants to raise the power beyond the Golden Section Search. These disadvantages of the three methods may be due to the approaches’ restrictive constant assumption of the signal structure. By contrast, our method – even when using only 10\% of the data – affords us the flexibility necessary for full-range detection. Additionally, the ROC curve is more impressive than for the three other methods. The single-selection results (i.e., green dots) significantly outperform the other methods as well, conveying the effectiveness of this bootstrap method in identifying a good threshold.

Finally, we compare the computational time for each instantiation. Parameter tuning is also included in the timing. Findings are listed in Table~\ref{tab:small-timing}.  When the full differential matrix is used, spectral projection is most efficient. The computational difference between our method and spectral projection solely involves tuning, as we have an extra step to select the rank. Yet this difference is negligible; both methods demonstrate comparable timing performance. The adaptive LAS is much slower than the other methods. However, if we use our bootstrapping approach to locate the best model, it takes longer: an average of  64.88 sec. and  80.14 sec. for Pearson’s and Spearman’s correlations, respectively. However, given its effectiveness in model selection, the bootstrap remains worthwhile as long as the problem size is not overly large. For large-scale problems, we still need to use compressed spectral screening to reduce the dimensionality first, as is evaluated next. Also, for this small-scale problem, the full version may be faster than the subsampled version purely due to implementation efficiency; that is, the subsampling and correlation calculation are implemented in R, but the default full correlation matrix calculation is implemented in C (as default functions in R).

\begin{table}[ht]
\centering
\caption{Computational time (in sec.) of four methods averaged over 50 replications with small-scale problems  ($n_1=n_2=100, p=2000$).}\label{tab:small-timing}
\begin{tabular}{l|cccc|cccc}
  \hline
 &  \multicolumn{4}{c|}{Pearson} &  \multicolumn{4}{c}{Spearman}  \\ 
  \hline
$\rho$ & 0.05 & 0.1 & 0.2& full &0.05 & 0.1 & 0.2 & full \\ 
  \hline
Our method & 0.652 & 1.30 & 2.54& 0.98 & 0.77 & 1.50 & 2.93 & 0.96\\ 
\hline
Spectral projection  & 0.64 & 1.23 & 2.33 & 0.61 & 0.74 &1.43 & 2.69 & 0.72 \\ 
\hline
Adaptive LAS&  \multicolumn{4}{c|}{2.58} &  \multicolumn{4}{c}{2.95}\\ 
\hline
Golden Section Search & \multicolumn{4}{c|}{1.17}  & \multicolumn{4}{c}{1.36}  \\ 
   \hline
\end{tabular}
\end{table}

\subsection{Results for large-scale problems}

 In this section, we evaluate our method for problems with a the size that more commonly seen in gene differential analysis. Specifically, we now evaluate our method in the scenario of $n_1=n_2=100$ and $p = 40000$. The scale of this problem makes our subsampling strategy necessary. Therefore, we focus on our CSS method and spectral projection based on our subsampling technique.

\begin{figure}[H]
\vspace{-0.2cm}
\centering
\begin{subfigure}[t]{0.42\textwidth}
\centering
\includegraphics[width=\textwidth]{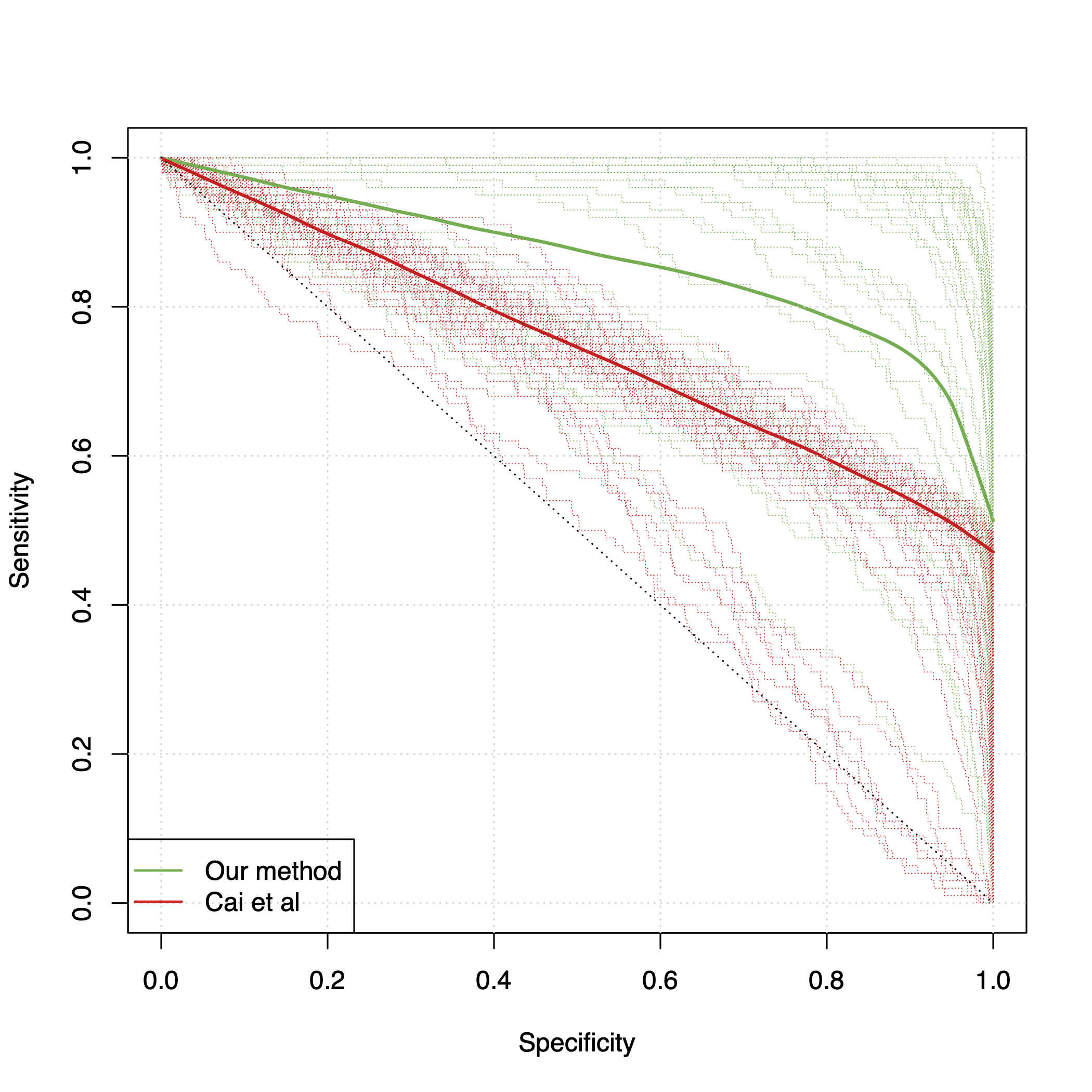}
\caption{Pearson's correlation with $\rho = 0.05$.}
\label{fig:large-Pearson-rho005}
\end{subfigure}
\hfill
\begin{subfigure}[t]{0.42\textwidth}
\centering
\includegraphics[width=\textwidth]{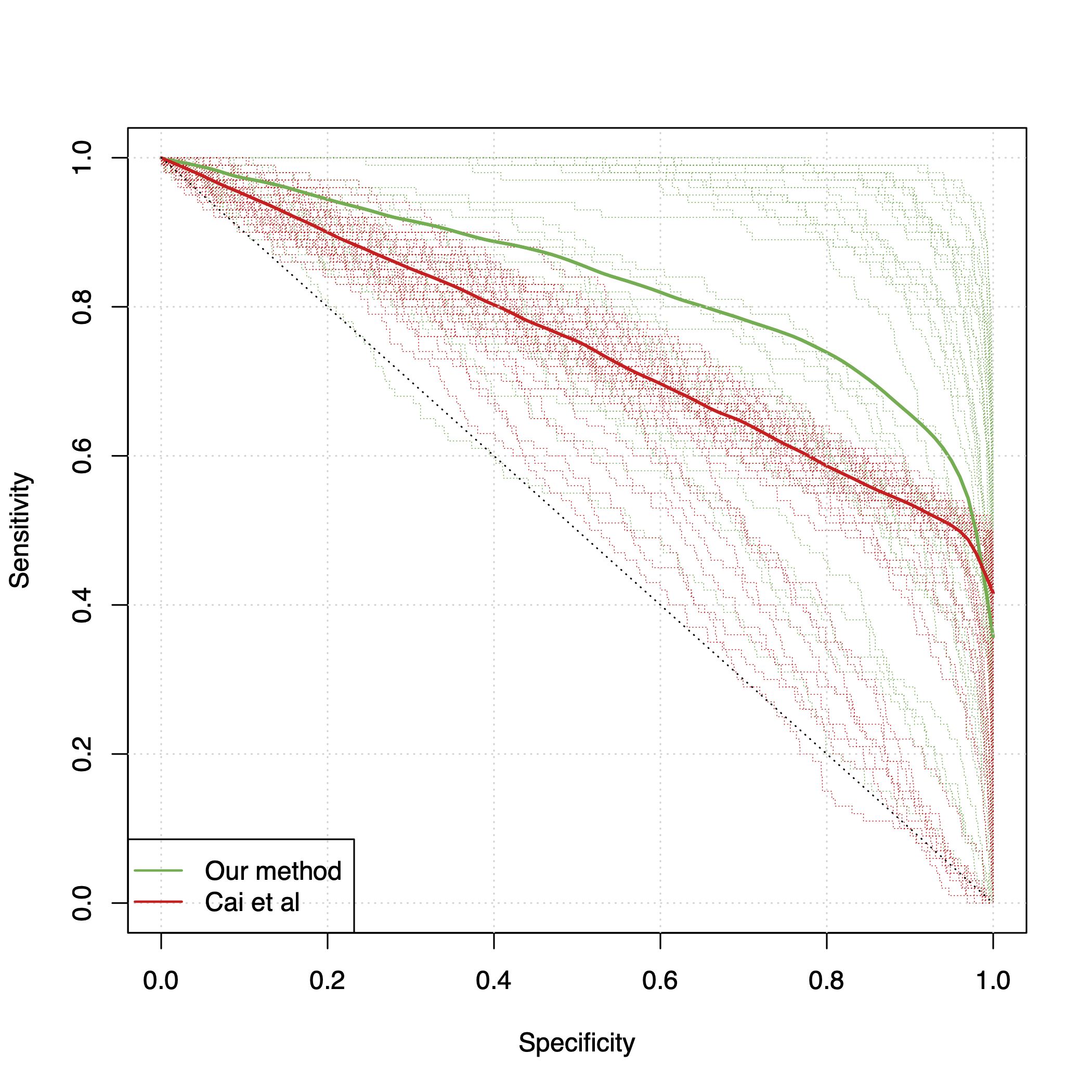}
\caption{Spearman's correlation with $\rho = 0.05$.}
\label{fig:large-Spearman-rho005}
\end{subfigure}

\begin{subfigure}[t]{0.42\textwidth}
\centering
\includegraphics[width=\textwidth]{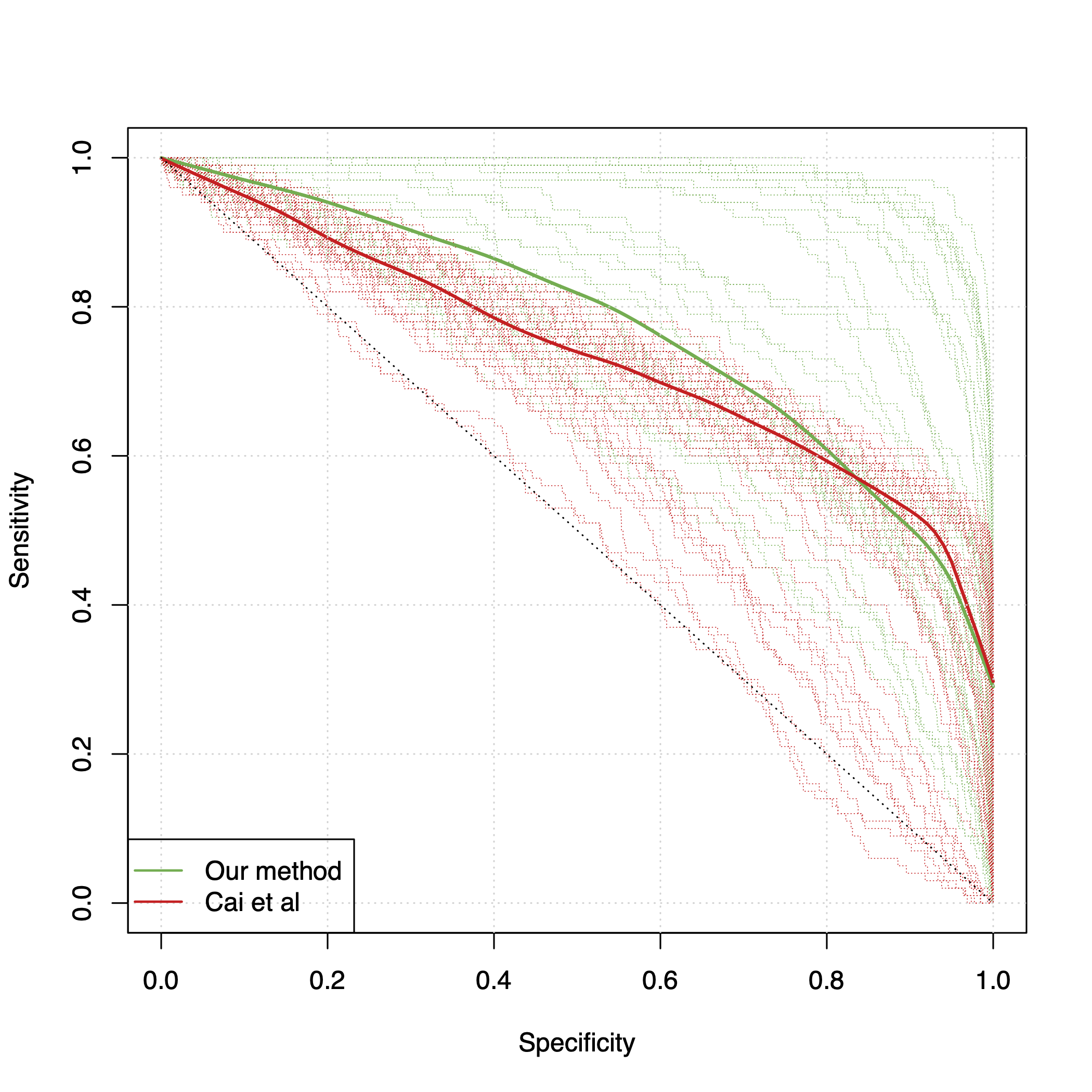}
\caption{Pearson's correlation with $\rho = 0.02$.}
\label{fig:large-Pearson-rho002}
\end{subfigure}
\hfill
\begin{subfigure}[t]{0.42\textwidth}
\centering
\includegraphics[width=\textwidth]{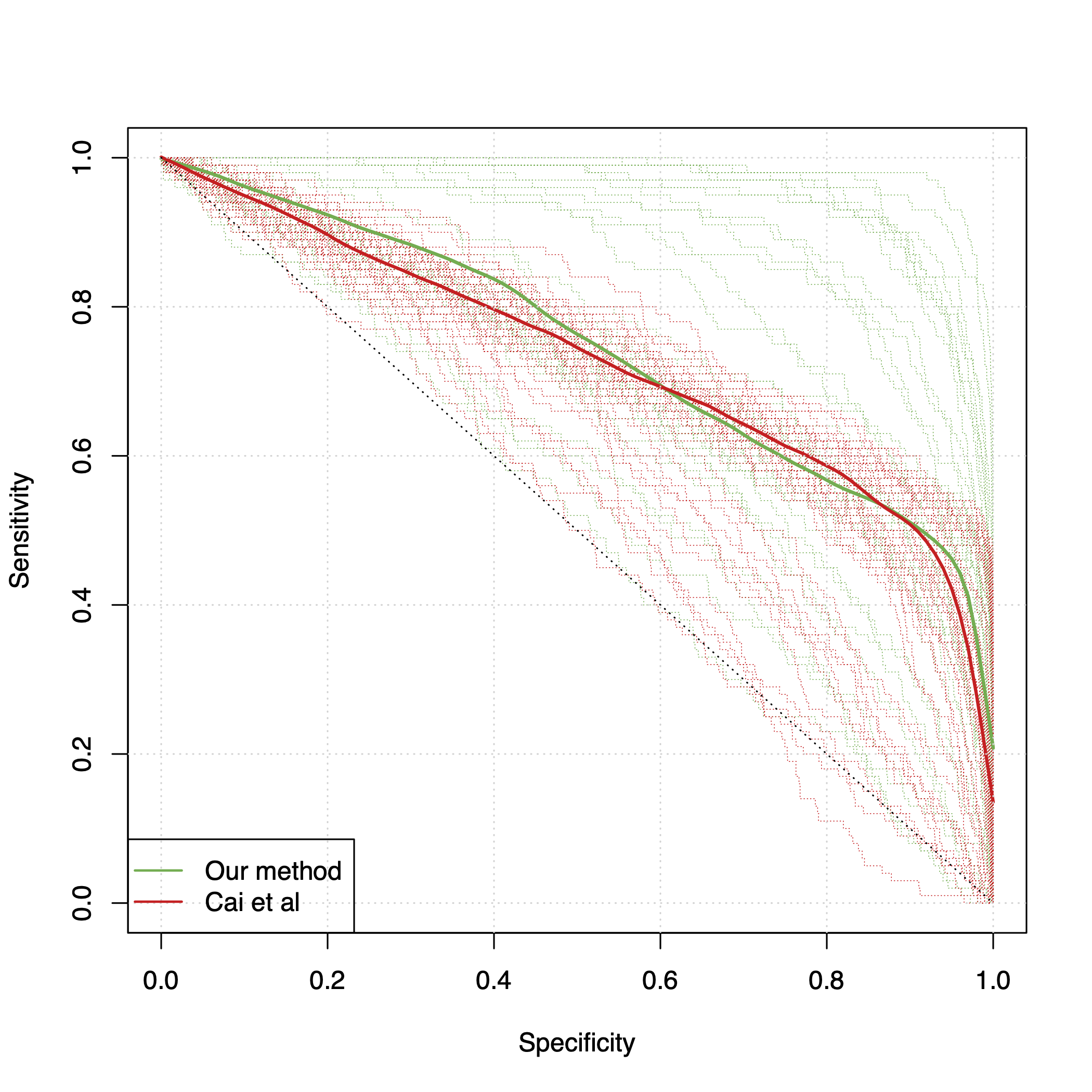}
\caption{Spearman's correlation with $\rho = 0.02$.}
\label{fig:large-Spearman-rho002}
\end{subfigure}

\begin{subfigure}[t]{0.42\textwidth}
\centering
\includegraphics[width=\textwidth]{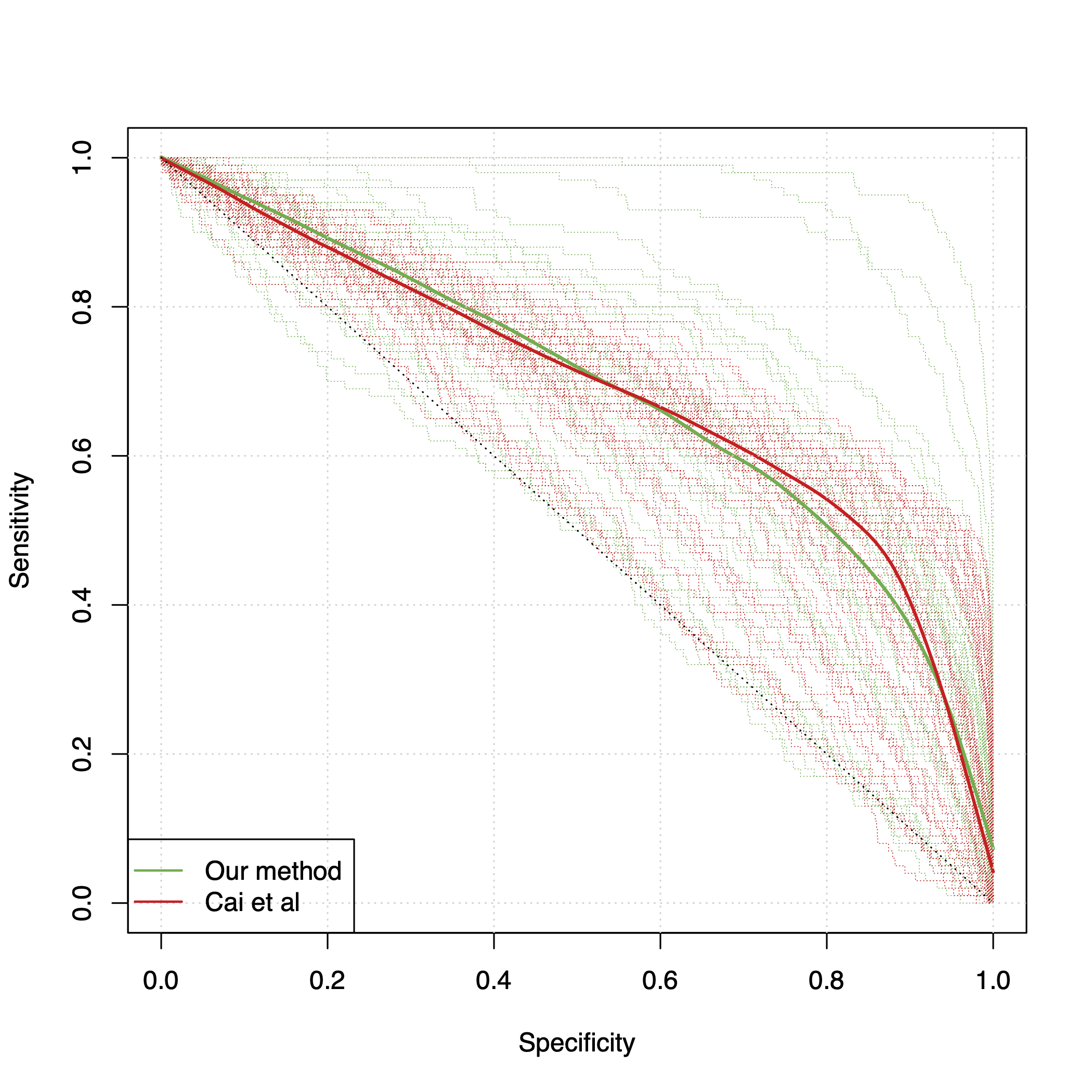}
\caption{Pearson's correlation with $\rho = 0.01$.}
\label{fig:large-Pearson-rho001}
\end{subfigure}
\hfill
\begin{subfigure}[t]{0.42\textwidth}
\centering
\includegraphics[width=\textwidth]{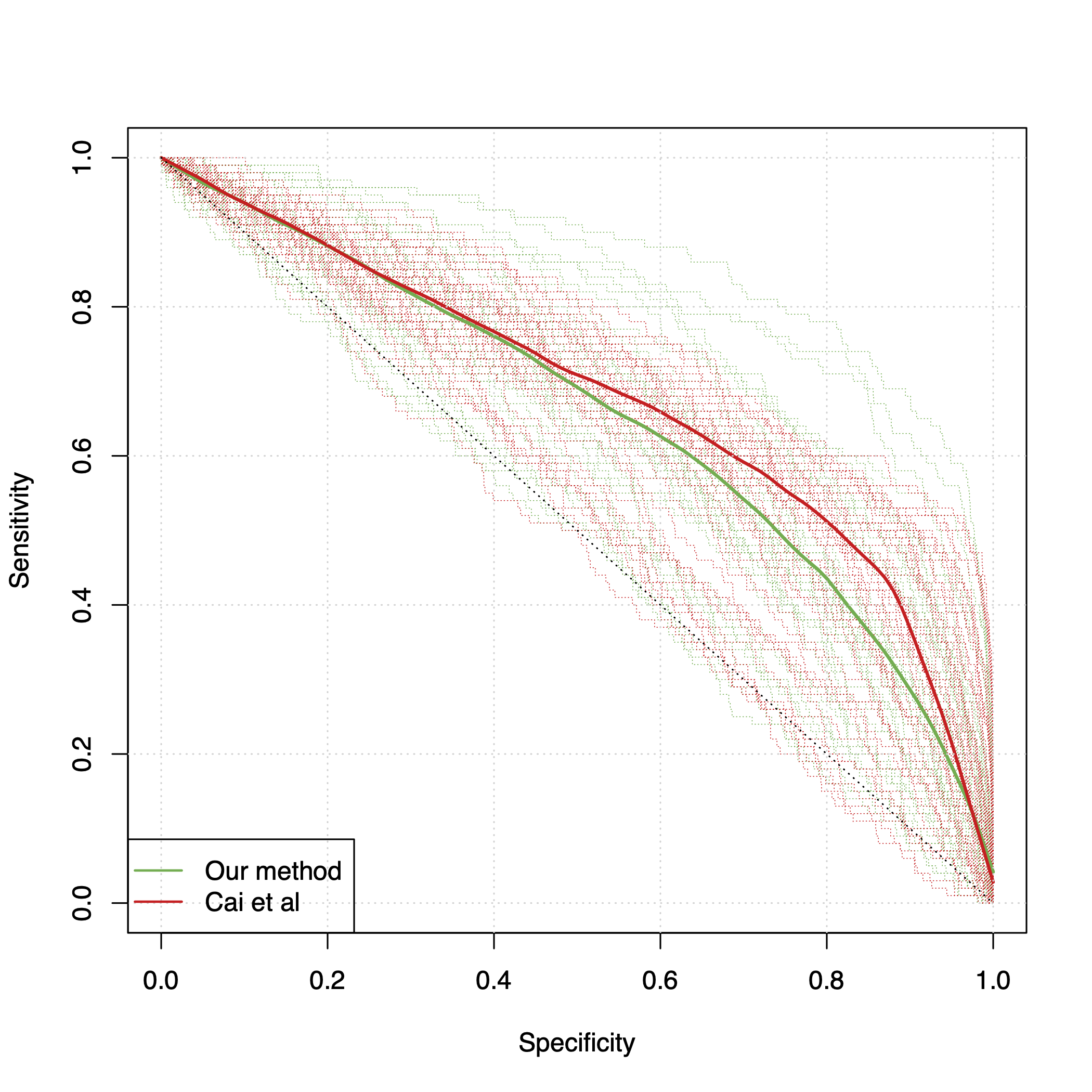}
\caption{Spearman's correlation with $\rho = 0.01$.}
\label{fig:large-Spearman-rho001}
\end{subfigure}
\caption{Sensitivity-specificity ROC plot for our method and the spectral projection method on subsampled correlation matrices in the large-scale problems ($n_1=n_2=100, p=40000$). }
\label{fig:large-ROC}
\end{figure}

Figure~\ref{fig:large-ROC} displays separate ROC curves for each configuration  $\rho=0.01, 0.02, 0.05$. When we use $5\%$ of the entries, the CSS performance is impressive --- much better than subsampled spectral projection. A nontrivial proportion of curves work nearly perfectly. When the sampling proportion drops to 2\%, the ROC curves become noisier. The advantage of the CSS over subsampled spectral projection degrades as well but still remains significant.  When $\rho$ declines further to $1\%$, the difference between the two methods becomes negligible. 

In practice, as addresed before, the choice of $\rho$ should depend on the computational capacity. A larger $\rho$ is always preferable. Our rule of thumb, $\rho > 2(n_1+n_2)/(p+1)$ gives $\rho > 0.01$ in this setting. This indicates having $\rho = 0.01$ is too low for an accurate estimation, matching our observation in Figure~\ref{fig:large-ROC}. Finally, to determine whether the effective performance at $\rho = 0.05$ is computationally feasible for the current problem, we include the timing results in Table~\ref{tab:large-timing}. The average computational time for our method with $\rho=0.05$ is about 438 seconds. The spectral projection is roughly 5\% faster but exhibits significantly lower accuracy (Figure~\ref{fig:large-Pearson-rho005} and \ref{fig:large-Spearman-rho005}). The configuration of $\rho=0.02$ is faster and gives reasonable accuracy as well. In our view, delivering such accurate screening results within 438 seconds is noteworthy given the size of the problem. The proposed method can therefore significantly expand the scalability of differential correlation analysis.

\begin{table}[ht]
\centering
\caption{The computational time (in sec.) of two methods based on subsampling averaged over the 50 replications in large-scale problems. }\label{tab:large-timing}
\begin{tabular}{l|ccc|ccc}
  \hline
 &  & Pearson & &  & Spearman &  \\ 
  \hline
$\rho$ & 0.01 & 0.02 & 0.05& 0.01 & 0.02 & 0.05 \\ 
  \hline
Our method & 83.038 & 158.152 & 438.684  & 74.474 & 142.945 & 390.495   \\ 
Spectral projection  & 81.689 & 153.878 & 419.787 & 73.266& 138.851 & 376.902  \\
   \hline
\end{tabular}
\end{table}

\section{Gene co-expression differential analysis for Glioblastoma}\label{sec:data-analysis}

Now we will demonstrate the compressed spectral screening method by applying it to analyze genes for glioblastoma. As mentioned, the original dataset contains 51,448 genes; 139 observations for the GBM group, 254 for the normal group, and 507 for the LGG group. After removing deficiently expressed genes (with a median expression level of less than 0.25) from the three groups, 23,296 genes remain for analysis. We focus on identifying a subset of genes with significant differences in their correlation matrices between the GBM patient group and the normal patient group (GTex). Our analysis is mainly based on the difference between the GBM group and the normal group, while the LGG group acts as the validation set.

 Due to potential skewness and the noisy nature of gene expression data, we apply logarithmic transformation and use Spearman’s correlation for our analysis. The sample is too large for direct correlation analysis. Therefore, we begin by applying compressed spectral screening to the 23,296 genes. In the tuning procedure, $K = 2$ is selected as the approximating rank. We select 2000 genes through this process to reduce the problem to a manageable size for all benchmark methods.

Next, we narrow our study to a smaller subset using the exact benchmark algorithms on the $2000\times 2000$ correlation matrices. Our spectral screening method effectively identifies a subgroup of genes. This pattern indicates effective detection of a subgroup of differential correlations with separably larger spectral scores, as shown in Figure~\ref{fig:SS-hist}. The other three benchmark methods fail to achieve a similar level of effectiveness for selection and downstream analysis (Appendix~\ref{sec:othermethods}). Based on the proposed bootstrap selection method, 254 genes appear to be differentially correlated.

Among the selected genes, the average correlation value is 0.628 for the GBM group, 0.41 for the LGG group, and 0.84 for the normal group. In addition to overall magnitude differences between the GBM and normal groups, we are interested in understanding the differential patterns in how genes are correlated. We analyze these data using network analysis on the differential networks, constructed by removing the differential effects of marginal magnitude. We next describe our analysis of the selected genes.

\begin{figure}[H]
\centering
\includegraphics[width=0.9\textwidth]{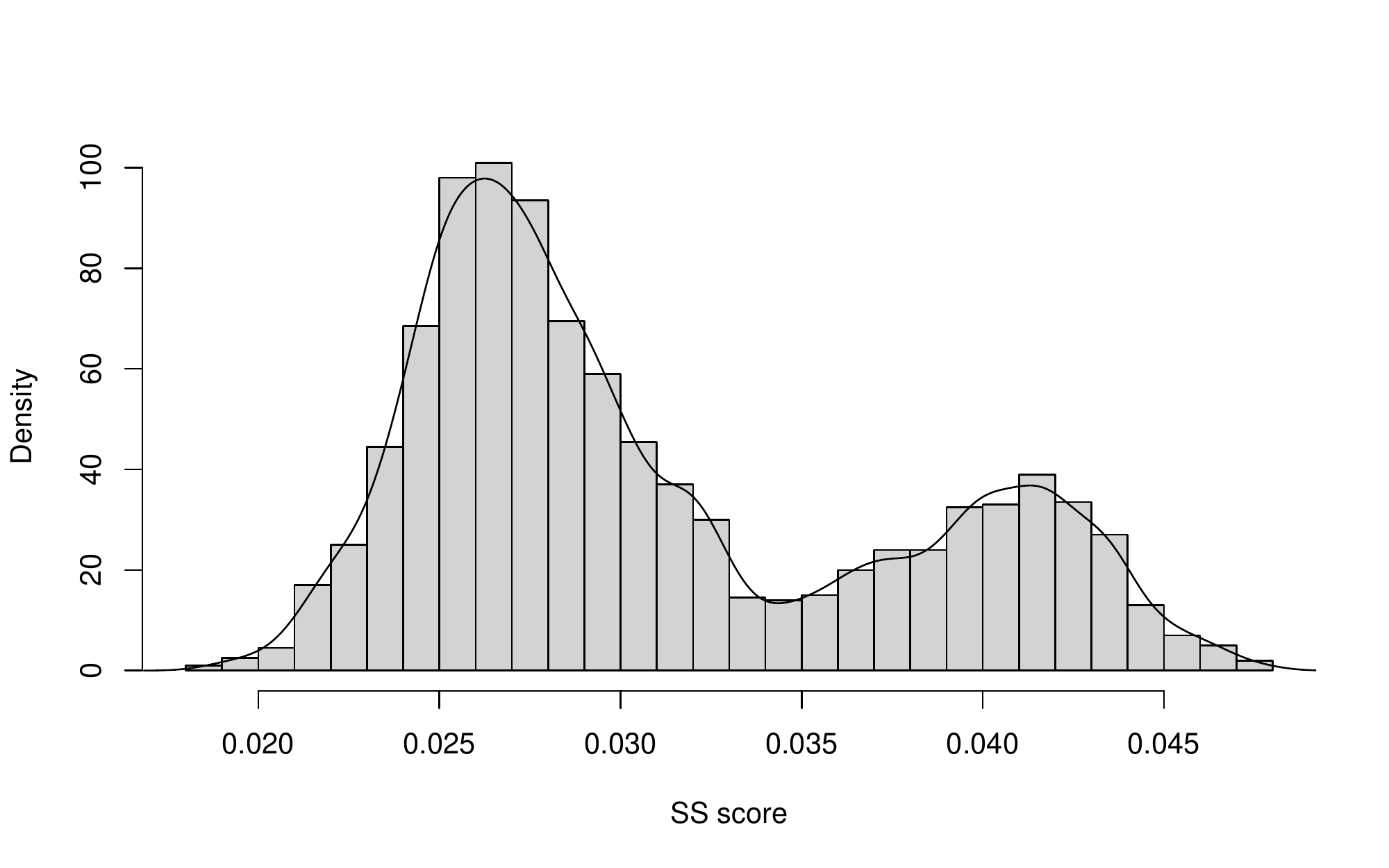}
\caption{The histogram and density of the spectral scores of 2000 genes by the spectral screening methods.}
\label{fig:SS-hist}
\end{figure}

Network analysis has been widely employed to evaluate biological and genetic relations \citep{gambardella2013differential}. One main advantage of using a network structure to represent correlations is that this data structure is potentially more robust and can remove magnitude effects. We first construct a network between genes according to Spearman’s correlation in the normal group by treating correlation entries larger than a threshold as an edge. The threshold is selected based on \cite{gambardella2013differential}'s strategy wherein the network degree best matches the power-law or scale-free distribution, which is generally believed to be a principled structure in biological networks \citep{barabasi2003scale,barabasi2009scale}. To ensure a meaningful comparison, we then select truncation thresholds for the GBM group and LGG group such that the resulting networks have the same density as the normal group. Finally, we select the largest connected component of the three networks, thereby obtaining three networks with 198 genes and an average degree of around 4.82. Magnitude effects have already been eliminated from these networks, and the structural differences between them only reflect relative differences in the networks’ connection patterns (Figure~\ref{fig:HCD-SS}). Next, we leverage network analysis methods to explore differential patterns between these networks.

 Given two networks on the same set of nodes, the networks can be represented by their adjacency matrices $A_1 \in \{0,1\}^{n\times n}$ and $A_2 \in \{0,1\}^{n\times n}$. Their differential adjacency can be defined as $A \in \{0,1\}^{n\times n}$ such that $A_{ij} = I(A_{1,ij}\ne A_{2,ij})$. The matrix $A$ The matrix A gives the dyad indices where the two networks differ. In particular, we take the differential adjacency matrix between the GBM network and the normal network. We want to summarize differential patterns based on multiple gene modules; that is, the differential patterns are similar for genes within the same module. We apply the hierarchical community detection method (HCD) of \cite{li2018hierarchical} is applied, which can automatically determine the module partition as well as the hierarchical relation between modules. Nodes’ module labels returned by the HCD algorithm are color-coded in  Figure~\ref{fig:HCD-SS}; the corresponding hierarchical relation between detected gene modules (communities) appears in Figure~\ref{fig:HCD-dendrogram}.

The modules and hierarchy provide informative distinctions regarding connection patterns. For example, in the normal group network, Modules 1 and 2 are densely intraconnected; both groups in the GBM network have sparse within-module connections. On the contrary, Module 3 is densely intraconnected in the GBM network but has few within-module connections in the normal network. Modules 1 and 2 thus exhibit different outgoing connection patterns. The connection between Modules 2 and 3 is much weaker in the GBM network than in the normal network, whereas this pattern does not hold for the connection between Modules 1 and 3. These distinct variation patterns consistently align with the hierarchical structure in Figure~\ref{fig:HCD-dendrogram}, indicating that Modules 1 and 2 may be treated as a meta-module with higher similarity compared with Module 3. Moreover, most notable changes in the connection pattern can be validated on the LGG network. To quantify this finding, in Table~\ref{tab:within-module}, we present the three networks’ within-module edge density. This density perfectly matches the normal-LGG-GBM order. As our analysis is performed solely based on the GBM and normal groups, this consistency offers compelling evidence for the effectiveness of our analysis.

\begin{figure}[H]
\centering
\includegraphics[width=\textwidth]{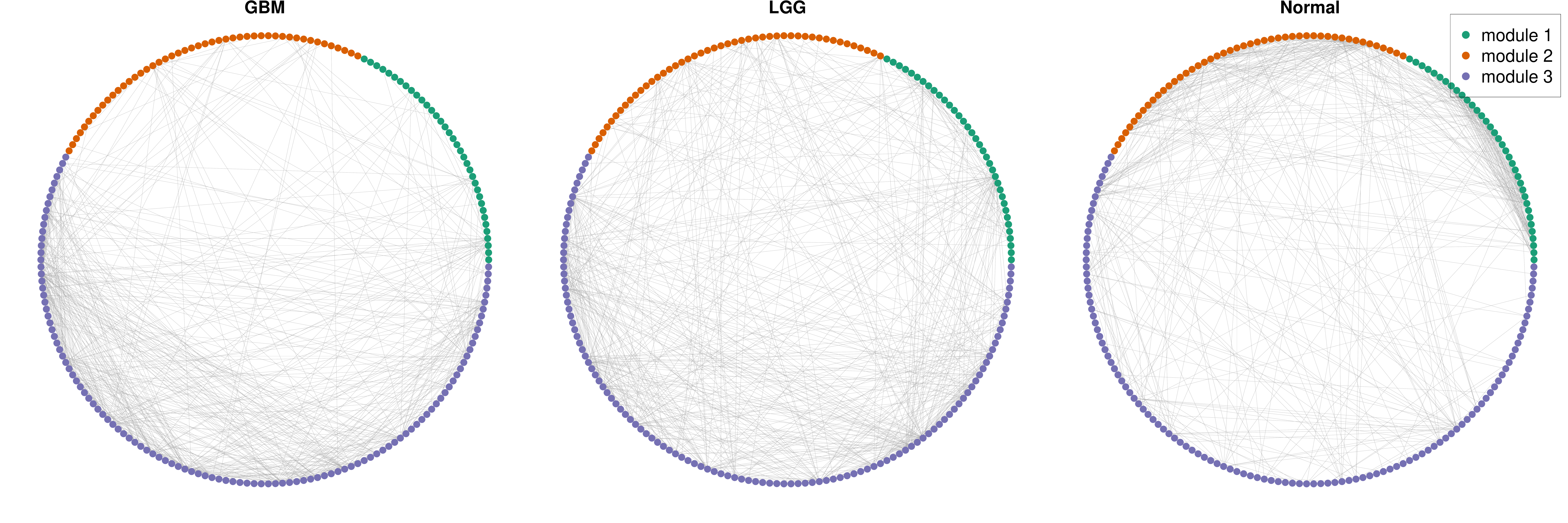}
\caption{Three networks of GBM, LGG, and normal groups and three gene modules identified by the HCD algorithm on the differential network between the GBM and normal groups.}
\label{fig:HCD-SS}
\end{figure}

\begin{figure}[H]
\centering
\vspace{-1cm}
\includegraphics[width=0.5\textwidth]{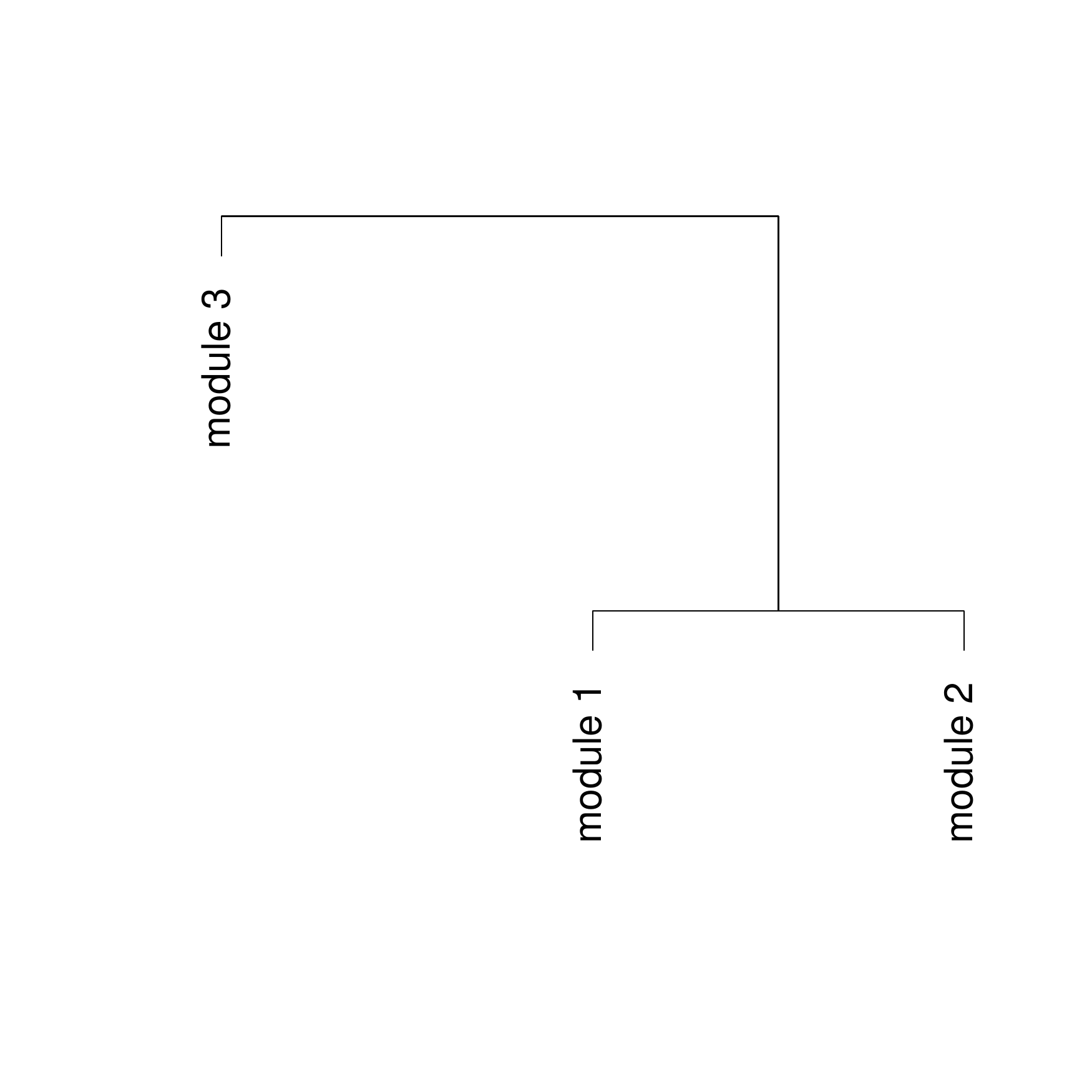}
\vspace{-1.5cm}
\caption{Dendrogram of hierarchical relation between 5 gene modules identified by the HCD algorithm on the differential network between the GBM and normal groups.}
\label{fig:HCD-dendrogram}
\end{figure}

\begin{table}[ht]
\centering
\caption{The within-module density of connections for the three networks.}
\label{tab:within-module}
\begin{tabular}{rrrr}
  \hline
 & GBM & LGG & Normal \\ 
  \hline
module 1 & 0.017 & 0.037 & 0.198 \\ 
 module 2 & 0.007 & 0.007 & 0.092 \\ 
 module 3 & 0.057 & 0.041 & 0.013 \\     \hline
\end{tabular}
\end{table}

The identified modules are also biologically meaningful. Specifically, we characterize each of the three gene modules via gene ontology analysis to identify their biological processes, molecular functions, and cellular components  \citep{mi2019panther}. The main results are summarized below.
\begin{itemize}
\item Module 1: The ontology analysis suggests that this module is enriched in genes exhibiting protein threonine kinase activity, protein serine kinase activity, and ATP binding function; and in genes coding for proteins with catalytic complex activity. Numerous genes coding for proteins with kinase activity have been shown to be overexpressed in gliomas; suppressing the activity of these overexpressed kinases is an area of active investigation \citep{el2021therapeutic,kai2020prognostic}. The low connectivity in GBM relative to normal samples is likely reflective of this dysregulation.
\item Module 2: : This module is enriched in genes involved in the regulation of chromatin organization, histone methylation, RNA localization, and transcription coregulator activity. Chromatin organization, histone methylation, and transcription factor activity have been found to differ substantially between glioma and normal samples \citep{bacolod2021mgmt,shi2020hotairm1,ciechomska2020histone,jia2020association}. Gene dysregulation in this module could contribute to aberrations in the processes that control gene expression, leading to the substantial differences in gene expression profiles seen between glioma and normal samples.
\item Module 3: Genes in this module are enriched for gene sets involving the regulation of RNA metabolic processing and nuclear cellular components. Naturally, nucleic acids and nuclear components are essential to the rapid cell division characteristic of Grade IV gliomas \citep{krupp2012rna}. It therefore seems logical that within-module connectivity is increased for the GBM group.
\end{itemize}
The detailed ontology enriched categories are presented in Appendix~\ref{sec:ontology}. 

\section{Theoretical properties}\label{sec:theory}

In this section, we present basic theoretical properties of the proposed method. We first study the  complete version of the spectral screening and show that the method can perfectly identify the differential variables under proper models. We call this a strong consistency property. In the second part, we study the compressed spectral clustering and show that at most a vanishing proportion of differential genes will be falsely removed in this case. We call this a weak consistency property. We start with introducing a few more notations for theoretical discussions. For a matrix $M$, we will use $M_{i\cdot}$ to denote its $i$th row and $M_{\cdot j}$ to denote its $j$th column. Define $\norm{M}_{\max} = \max_{ij}|M_{ij}|$, 
$\norm{M}_{\infty} = \max_{i}\norm{M_{i\cdot}}_1 =  \max_{i}\sum_j|M_{ij}|$, and $
\norm{M}_{2,\infty} = \max_{i}\norm{M_{i\cdot}}_2$.
For a positive semidefinite matrix $M$, define its stable rank as $\tau(M) = \tr(M)/\norm{M}$ and note that $\tau(M) \le \Rank(M)$.

\begin{ass}\label{ass:model}
 Assume $x_1, \cdots, x_{n_1}$ be i.i.d sub-gaussian random vectors with zero mean and covariance $\Sigma_1$ and $y_1, \cdots, y_{n_2}$ be i.i.d random vectors with zero mean and covariance $\Sigma_2$.  Moreover,  $\max(\norm{x_i}_{\psi_2},\norm{y_i}_{\psi_2}) \le \sigma^2$ for some constant $\sigma^2>0$ where $\norm{}_{\psi_2}$ is the Orlicz norm for sub-gaussian distributions. In particular, this is equivalent to the definition that for any $v \in \bR^{p}$ with $\norm{v}_2 = 1$,
 $$\p(|x_i^Tv|>t) \le 2e^{-t^2/\sigma^2}$$
 for any $t>0$, $1\le i \le n_1$. And the same property holds for each $y_i$ as well. 
\end{ass}

As mentioned, our main structural assumption is that $D=\Sigma_1 - \Sigma_2$ has zero entries except within a diagonal blocks. Without loss of generality, we take the first $m$ genes as differential variables. In additional, we need $D$ to be low-rank. We formulate these intuitions in the following assumption.

\begin{ass}\label{ass:difference}
Assume the rank $K$ matrix $D = \Sigma_1-\Sigma_2$ has the following eigen-decomposition $D= U\Lambda U^T$ such that $\Lambda = \diag(\lambda_1, \cdots, \lambda_K)$ and $U\in\bR^{p\times K}$ is an orthogonal matrix. It is also assumed that $|\lambda_1| \ge |\lambda_2| \ge \cdots, \ge |\lambda_K| >0$. Furthermore, assume $\norm{U_{i\cdot}}>0, i\le m$ and $\norm{U_{i\cdot}} = 0$ for $i>m$.
\end{ass}

In addition, we will need sufficient regularity for the problem to be solvable. A common property involved in spectral methods is the so-called \emph{incoherence assumption}, originally introduced by \cite{candes2009exact}. It is believed to be necessary to control entrywise errors in matrix recovery \citep{fan2018eigenvector,abbe2017entrywise,cape2019two}. In our situation, we would have $U$ to be zero in most rows, so the standard incoherence assumption is not meaningful. Instead, we constrain the assumption on only the nonzero rows corresponding to the differential variables with a slightly stronger requirement. Essentially, we are assuming the none of the differential variables has negligible magnitude of $U_{i\cdot}$ compared to others.

\begin{ass}[Constrained incoherence]\label{ass:incoherence}
Under the Assumption~\ref{ass:difference}, there exists a constant $\mu >1$ such that
$$\frac{1}{\mu} \frac{K}{m} \le \min_{1\le i \le m} \norm{U_{i\cdot}}^2 \le  \max_{1\le i \le m} \norm{U_{i\cdot}}^2= \norm{U}_{2,\infty}^2 \le \mu\frac{K}{m}.$$
For notational simplicity, we define $\Delta = \min_{1\le i \le m} \norm{U_{i\cdot}}/4$.
\end{ass}

The next theorem shows that the spectral screening on the complete data set is guaranteed to separate the differential variables from the rest according to the spectral scores. 

\begin{thm}[Strong consistency of the complete spectral screening]\label{thm:strong-consistency}
Under Assumptions~\ref{ass:model}, \ref{ass:difference} and \ref{ass:incoherence}, let $\hat{U}$ be the vector of the top $r$ leading eigenvectors of the matrix $\hat{D} = \hat{\Sigma}_1-\hat{\Sigma}_2$, where $\hat{\Sigma}_1$ and $\hat{\Sigma}_2$ be the sample covariance matrices of $\{x_i\}_{i=1}^{n_1}$ and $\{y_i\}_{i=1}^{n_2}$, respectively. For some constant $c>0$, there exists a constant $C$ depending only on $K, c$ and $\mu$, such that if  for sufficiently large $n$
\begin{equation}\label{eq:signal-bound}
C\sigma^2\frac{p}{|\lambda_K|}\sqrt{\frac{\log{p}}{n}} < 1
\end{equation}
we have
$$\min_{1\le i\le m}\norm{\hat{U}_{i\cdot}}_2 > \max_{i > m}\norm{\hat{U}_{i\cdot}}_2 + 2\Delta$$
with probability at least $1-2n^{-c}$, where $n = \min(n_1, n_2)$.
\end{thm}

Next, we introduce the theoretical guarantee for the compressed spectral screening, based on randomly subsampling the entries of $\hat{D}$. Intuitively, since we use only a subset of the data, we may not be able to achieve exactly the same level of accuracy as the complete data set. However, since we are using a randomly sampled subset, our result should still be correct in the average sense. This intuition can be formally justified by the following weak consistency of the compressed spectral screening. We will state this property in the number of differential variables that are not among the top $m$ ones with the largest spectral scores.
\begin{defi}
Let $\tilde{U}$ be the top $r$ eigenvector matrix of $\tilde{D}$ in Algorithm~\ref{algo:css}. Let $\eta_{\tilde{U}}(i)$ be the rank of $\norm{\tilde{U}_{i\cdot}}$ among all the $p$ rows in decreasing order.  Define the confusion count of the spectral screening to be
$$q_{\tilde{U}} = |\{i: 1 \le i \le m, \eta_{\tilde{U}}(i) > m\}|.$$
\end{defi}

\begin{thm}[Weak consistency of the compressed spectral screening]\label{thm:weak-consistency}
Under Assumptions~\ref{ass:model}, \ref{ass:difference} and \ref{ass:incoherence}, assume $\rho \ge \log{p}/p$. For any some $c>0$, there exists a constant $C_1$ depending on $c, \mu$ and $K$, such that  for a sufficiently large $n$
if 
$$\tau(\Sigma_1)\log{p} < n_1, \tau(\Sigma_2)\log{p} < n_2$$
and
\begin{equation}\label{eq:condition2}
|\lambda_K| > C_1\left(  \sqrt{\frac{\tau(\Sigma_1)\log{p}+\log{n_1}}{n_1}}\norm{\Sigma_1} + \sqrt{\frac{\tau(\Sigma_2)\log{p}+\log{n_2}}{n_2}}\norm{\Sigma_2}\right),
\end{equation}
we have
$$q_{\tilde{U}}  \le C\sigma^4\frac{p(n_1+n_2)}{\rho K |\lambda_K|^2}m$$
with probability at least $1-4n^{-c}-p^{-c}$, where $n=\min(n_1,n_2)$.
\end{thm}

Notice that there is a gap between Theorem~\ref{thm:strong-consistency} and Theorem~\ref{thm:weak-consistency}, because even if $\rho=1$ in Theorem~\ref{thm:weak-consistency}, the result is still weaker than  Theorem~\ref{thm:strong-consistency}. We believe this is an artifact of our proving strategy and will leave the refinement for future work.

\medskip

We now discuss the indications of Theorem~\ref{thm:strong-consistency} and Theorem~\ref{thm:weak-consistency} in a simplified case where the results more interpretable. In particular, we consider a special case of the spiked covariance model \eqref{eq:spikedmodel}:
$$\Sigma_1 = \Sigma_2+\lambda_1 u_1u_1^T$$
with $n_1=n_2 = n$ and $u_1$ satisfying the incoherence Assumption \ref{ass:incoherence}. Assume that $\tau(\Sigma_2) \le n/\log{p}$ and $\max_{jj}\Sigma_{2,jj} = 1$. In this case, we can see that we can set $\sigma^2 = C\max(1, 1+|\lambda_1|\mu/m)$ following Assumption~\ref{ass:difference}. Therefore, Theorem~\ref{thm:strong-consistency} indicates that if
$$p\sqrt{\frac{\log p}{n}} \ll |\lambda_1| \le m/\mu$$
the strong consistency can be achieved.  For example, when $m = \kappa p$ for some constant $\kappa$, we require $\log p \ll n$ to have a nontrivial regime for $\lambda_1$. Similar, for the weak consistency, we need
$|\lambda_1|/m = O(1)$. Combining this requirement with the one for weak consistency, a necessary condition is
$$\sqrt{pn/\rho} = o(|\lambda_1|) = o(m) = o(p).$$
Therefore, we need to ensure $\rho \gg n/p$, which coincides with our rule of thumb about the lower bound of $\rho$ in Section~\ref{sec:tune}. Empirically, this selection give reasonable $\rho$ range according to our experiments.

\section{Discussion}\label{sec:discussion}
This paper proposes spectral screening algorithms to identify variables with differential patterns between two covariance/correlation matrices based on the need to assess a large-scale gene expression dataset related to glioblastoma. Our method assumes that the differences are constrained within a block of differential matrices but makes no additional assumptions about differential patterns. By leveraging spectral properties, we can incorporate random sampling to significantly boost the computational efficiency, such that this method can easily handle analyses of high-dimensional covariance matrices that are not even readable by computer memory. We have demonstrated the effectiveness of our spectral screening method in terms of variable selection accuracy and computational efficiency. A detailed differential co-expression study of TCGA and GTeX data demonstrates how the spectral screening method can clarify glioblastoma mechanisms. Our method is also well suited to a much broader range of applications, and its performance is theoretically guaranteed.

Of note, this spectral screening method can identify differential genes but cannot provide further insight into specific differential patterns. More in-depth studies of differential patterns require downstream analysis. In our glioblastoma example, we use a hierarchical community detection algorithm from \cite{li2018hierarchical} to extract precise differential patterns. The properties of this step are unknown with respect to differential correlation settings. As such, it would be useful to incorporate these two analytical steps into one systematic modeling procedure given that both are based on the spectral properties of data. The design and theoretical study of such a method presents an intriguing avenue for future work.

\section*{Acknowledgement}
T. Li is supported in part by the
NSF grant DMS-2015298 and the 3-Caveliers award from the University of Virginia. X. Tang is supported in part by the NSF grant DMS-2113467 and the EIM grant from  the University of Virginia.    A. Chatrath is supported by the NIH grant T32 GM007267.

\bibliography{CommonBib}{}
\bibliographystyle{abbrvnat}

\newpage

\begin{appendix}
\section{Proofs}

\begin{proof}[Proof of Theorem~\ref{thm:strong-consistency}]
Under Assumption~\ref{ass:model}, by Lemma 1 in \cite{ravikumar2011high}, we have
\begin{equation}\label{eq:event1}
\max\left( \norm{\hat{\Sigma}_1-\Sigma_1}_{\max}, \norm{\hat{\Sigma}_2-\Sigma_2}_{\max}\right) \le \sigma^2 \sqrt{\frac{2c\log{p}}{n}}
\end{equation}
with probability at least $1-2n^{-c}$. Under this event, by Lemma~\ref{lem:cape}, we have some orthogonal matrix $O \in \bR^{r\times r}$ such that
\begin{align}\label{eq:bound1}
\norm{\hat{U}-UO}_{2,\infty} &\le \frac{14}{|\lambda_K|}\norm{\hat{D}-D}_{\infty}\norm{U}_{2,\infty}. \notag\\
& \le \frac{14}{|\lambda_K|}\norm{U}_{2,\infty} \left( \norm{\hat{\Sigma}_1-\Sigma_1}_{\infty}+\norm{\hat{\Sigma}_2-\Sigma_2}_{\infty}\right) \notag\\
& \le \frac{14}{|\lambda_K|}\norm{U}_{2,\infty} p\left( \norm{\hat{\Sigma}_1-\Sigma_1}_{\max}+\norm{\hat{\Sigma}_2-\Sigma_2}_{max}\right) \notag \\
& \le 28\sqrt{2c}\sigma^2\frac{p}{|\lambda_K|}\sqrt{\frac{\log{p}}{n}}\norm{U}_{2,\infty}, 
\end{align}
as long as
\begin{equation}\label{eq:check1}
|\lambda_K| \ge 4\norm{\hat{D}-D}_{\infty} \ge 8\sqrt{2c}\sigma^2p\sqrt{\frac{\log{p}}{n}}.
\end{equation}
Notice that the zero pattern and the norms of rows do not change from $U$ to $UO$ for any orthogonal matrix $O$. Therefore, we will not distinguish $U$ and $UO$ for notational simplicity.  Using Assumption~\ref{ass:incoherence}, \eqref{eq:bound1} indicates
$$\norm{\hat{U}-U}_{2,\infty} \le 28\sqrt{2c}\sigma^2\frac{p}{|\lambda_K|}\sqrt{\frac{\log{p}}{n}}\mu \min_{1\le i\le m}\norm{U_{i\cdot}} = 112\mu\sqrt{2c}\sigma^2\frac{p}{|\lambda_K|}\sqrt{\frac{\log{p}}{n}}\Delta. $$

In particular, when $112\mu\sqrt{2c}\sigma^2\frac{p}{|\lambda_K|}\sqrt{\frac{\log{p}}{n}} < 1$, which indicates \eqref{eq:check1}, we have
$$|\norm{\hat{U}_{i\cdot}} - \norm{U_{i\cdot}}| \le \norm{\hat{U}_{i\cdot}-U_{i\cdot}} \le \norm{\hat{U}-U}_{2,\infty} <\Delta$$
because of \eqref{eq:signal-bound}. The claim of the theorem directly follows.
\end{proof}

\begin{proof}[Proof of Theorem~\ref{thm:weak-consistency}]
Let $\tilde{D} = \frac{1}{\rho} \hat{D}\circ \Omega$ where $\Omega$ is the matrix with $\Omega_{ij}$ being the sampled indicator. Let $\tilde{U}$ be the top $r$ eigenvectors of $\tilde{D}$, according to decreasing order (in magnitude) of its eigenvalues. By Davis-Kahan theorem \citep{yu2015useful}, there exits an orthogonal matrix $\hat{O} \in \bR^{r\times r}$ such that
\begin{equation}\label{eq:davis-kahan}
\norm{\tilde{U}-\hat{U}\hat{O}}_F^2 \le C\frac{\norm{\tilde{D}-\hat{D}}^2}{|\hat{\lambda}_K - \hat{\lambda}_{r+1}|^2}
\end{equation}
where $C$ is a constant depending on $r$. We now proceed to bound both the nominator and denominator. Firstly, notice that $\Omega$ can be seen as the adjacency matrix of an Erd\"{o}s-Renyi random graph with edge probability $\rho$. Because of the independence between $\Omega$ and our data, we can treat $\hat{D}$ as fixed when calculating the probability related to the randomness of $\Omega$. By Lemma~\ref{lemma:li} and Assumption~\ref{ass:model}, we have
\begin{align*}
\norm{\tilde{D}-\hat{D}}^2 &\le C'(\sqrt{\frac{p(n_1+n_2)}{\rho}}\norm{\hat{D}}_{\max})^2\\
& \le C'(\sqrt{\frac{p(n_1+n_2)}{\rho}}(\norm{D}_{\max}+\norm{\hat{D}-D}_{\max})^2\\
& \le C'\frac{p(n_1+n_2)}{\rho}(2\sigma^2 +\norm{\hat{D}-D}_{\max})^2
\end{align*}
with probability at least $1-p^{-c}$. Consider the joint event of the above and the event of \eqref{eq:event1}, and assume $\sqrt{\frac{2c\log{p}}{n}} < 1$, we have
\begin{equation}\label{eq:bound-nom}
\norm{\tilde{D}-\hat{D}}^2 \le C'\frac{p(n_1+n_2)}{\rho}(2\sigma^2 +2\sigma^2 \sqrt{\frac{2c\log{p}}{n}})^2\le 16C'\sigma^4\frac{p(n_1+n_2)}{\rho}
\end{equation}
with probability at least $1-2n^{-c}-p^{-c}$.

Secondly, by Weyl's theorem, we also have
$$|\hat{\lambda}_K - \hat{\lambda}_{r+1}| \ge |\lambda_K| - \norm{\hat{D}-D} \ge |\lambda_K|-\norm{\hat{\Sigma}_1-\Sigma_1} -\norm{\hat{\Sigma}_2-\Sigma_2}.$$
Standard spectral analysis of sub-gaussian random vectors (e.g. \cite{vershynin2018high}) indicates that
$$\norm{\hat{\Sigma}_1-\Sigma_1} \le C''\sqrt{\frac{\tau(\Sigma_1)\log{p}+\log{n_1}}{n_1}}\norm{\Sigma_1}$$
and
$$\norm{\hat{\Sigma}_2-\Sigma_2} \le C'' \sqrt{\frac{\tau(\Sigma_2)\log{p}+\log{n_2}}{n_2}}\norm{\Sigma_2}$$
with probability at least $1-2n^{-c}$, in which $\tau(\Sigma_1) = \tr(\Sigma_1)/\norm{\Sigma_1}$ is the stable rank. Combining the above results with \eqref{eq:davis-kahan} and if
\begin{equation}\label{eq:check2}
|\lambda_K| > 2C''\left(  \sqrt{\frac{\tau(\Sigma_1)\log{p}+\log{n_1}}{n_1}}\norm{\Sigma_1} + \sqrt{\frac{\tau(\Sigma_2)\log{p}+\log{n_2}}{n_2}}\norm{\Sigma_2}\right),
\end{equation}
we have
\begin{align}\label{eq:Frobenius-bound}
\norm{\tilde{U}-\hat{U}\hat{O}}_F^2 &\le C\frac{\norm{\tilde{D}-\hat{D}}^2}{|\hat{\lambda}_K - \hat{\lambda}_{r+1}|^2}\notag \\
&\le C16C'\sigma^4\frac{p(n_1+n_2)}{\rho}\frac{1}{( |\lambda_K|-\norm{\hat{\Sigma}_1-\Sigma_1} -\norm{\hat{\Sigma}_2-\Sigma_2})^2}\notag \\
&\le C16C'\sigma^4\frac{p(n_1+n_2)}{\rho}\frac{1}{( |\lambda_K|/2)^2}\notag \\
&\le C'''\sigma^4\frac{p(n_1+n_2)}{\rho|\lambda_K|^2}
\end{align}
with probability at least $1-4n^{-c}-p^{-c}$.
\medskip

Note that every differential variable $i \in [1,m]$ falling out of the first $m$ in $\norm{\tilde{U}_{i\cdot}}$ indicates one null variable enters the top $m$. Without lost of generality, assume the first $1\cdots q_{\tilde{U}}$ variables fall out of the top $m$ in the spectral scores. So we can match $q_{\tilde{U}}$ of these differential variables with $q_{\tilde{U}}$ null variables (this mapping is not unique). Denote this one-to-one mapping by $i \to t(i)$.  Under the current event, according to the proof of Theorem~\ref{thm:strong-consistency}, we have
 $$\norm{\hat{U}_{i\cdot}}_2 > \norm{\hat{U}_{t(i)\cdot}}_2 + 2\Delta. $$
 Therefore, we know that for each $i = 1, \cdots, q_{\tilde{U}}$, at least one of the following two
 $$|\norm{\hat{U}_{i\cdot}}_2- \norm{\tilde{U}_{i\cdot}}_2| > \Delta \text{~~and~~}|\norm{\hat{U}_{t(i)\cdot}}_2- \norm{\tilde{U}_{t(i)\cdot}}_2| > \Delta$$
  hold. Therefore, we have
 \begin{align*}
 \norm{\tilde{U}-\hat{U}\hat{O}}_F^2 & = \sum_{i} \norm{\tilde{U}_{i\cdot} - \hat{U}_{i\cdot}\hat{O}}_2^2 \ge \sum_{i}|\norm{\tilde{U}_{i\cdot}}_2-\norm{\hat{U}_{i\cdot}\hat{O}}_2|^ \ge \sum_{i}|\norm{\tilde{U}_{i\cdot}}_2 - \norm{\hat{U}_{i\cdot}}_2|^2 \ge  q_{\tilde{U}}\Delta^2.
 \end{align*}

Using \eqref{eq:Frobenius-bound} and Assumption~\ref{ass:incoherence}, we get
$$q_{\tilde{U}} \le \frac{\norm{\tilde{U}-\hat{U}\hat{O}}_F^2}{\frac{1}{16\mu} \frac{K}{m}} \le C\sigma^4\frac{p(n_1+n_2)}{\rho K |\lambda_K|^2}m$$
as claimed.
\end{proof}

\begin{lem}[Theorem 4.2 of \cite{cape2019two}]\label{lem:cape}
Let $X$ and $E$ be two $p\times p$ symmetric matrices and $\Rank(X) = r$. If the eigen-decomposition of $X$ is given by $X = U\Lambda U^T$ with $\Lambda$ including the eigenvalues in decreasing order (in magnitude). Further, let $\hat{U}$ be the corresponding top $r$ eigenvectors of $X+E$. If $|\Lambda_{rr}| \ge 4\norm{E}_{\infty}$, there exists an orthogonal matrix $O \in \bR^{r\times r}$ such that
$$\norm{\hat{U} - UO}_{2,\infty} \le 14\frac{\norm{E}_{\infty}}{|\Lambda_{rr}|}\norm{U}_{2,\infty}.$$
\end{lem}

Let $G \in \{0,1\}^{n\times n}$ be an adjacency matrix of an Erd\"{o}s-Renyi random graph where all edges appear independently with probability $\rho$ and $Z\in \bR^{n\times n}$ be a symmetric matrix.  Let $Z\circ G$ be the Hadamard (element-wise) matrix product of the two matrices.
\begin{lem}[Lemma 2 of \cite{li2016network}]\label{lemma:li}
Let $G \in \{0,1\}^{p\times p}$ be a $p\times p$ adjacency matrix of an Erd\"{o}s-Renyi random graph, where upper triangular entries are generated by independent Bernoulli distribution with expectation $\rho$. If $\rho \ge C\log p/p$ for a constant $C$. Then for any $c > 0$ and for any fixed matrix $Z \in \bR^{p\times p}$ with $\Rank(Z) \le K$, the following relationship holds 
$$\left\| \frac{1}{\rho}Z \circ G - Z \right\|  \le 2C'\sqrt{\frac{pK}{\rho}}\norm{Z}_{\max} $$
with probability at least $1-p^{-c}$,  where $C' = C'(c, C)$.
\end{lem}

\newpage 
\section{Ontology results about hierarchical gene modules}\label{sec:ontology}
\begin{table}[H]
{\tiny
\caption{Gene ontology results for gene sets enriched with the genes from the three modules identified by the HCD algorithm. The adjusted p-value is done corrected by the FDR control procedure of \cite{benjamini1995controlling}.}
\centering
\begin{tabular}{rlcrr}
  \hline
Module & Annotation Data Set & Gene Set & raw p-val & BH p-val \\ 
  \hline
  1 & GO Molecular Function & Protein Threonine Kinase Activity & $<0.01$ & 0.05 \\ 
    1 & GO Molecular Function & Protein Serine Kinase Activity & $<0.01$ & 0.05 \\ 
    1 & GO Molecular Function & ATP Binding & $<0.01$ & 0.12 \\ 
    1 & GO Cellular Component & Catalytic Complex & $<0.01$ & 0.03 \\ 
    1 & GO Cellular Component & Intracellular Non-membrane Bounded Organelle & $<0.01$ & 0.04 \\ 
    2 & GO Biological Process & chromatin organization  & $<0.01$ & $<0.01$ \\ 
    2 & GO Biological Process & regulation of chromatin organization & $<0.01$ & $<0.01$ \\ 
    2 & GO Biological Process & regulation of histone methylation  & $<0.01$ & $<0.01$ \\ 
    2 & GO Biological Process & establishment of RNA localization  & $<0.01$ & $<0.01$ \\ 
    2 & GO Biological Process & RNA localization  & $<0.01$ & $<0.01$ \\ 
    2 & GO Molecular Function & Transcription Coregulator Activity & $<0.01$ & $<0.01$ \\ 
    2 & GO Molecular Function & Chromatin Binding & $<0.01$ & $<0.01$ \\ 
    2 & GO Cellular Component & Nuclear Protein-Containing Complex & $<0.01$ & $<0.01$ \\ 
    2 & GO Cellular Component & Chromosome & $<0.01$ & 0.01 \\ 
    2 & GO Cellular Component & Nucleosome & $<0.01$ & $<0.01$ \\ 
    3 & GO Biological Process & regulation of transcription by RNA polymerase II  & $<0.01$ & $<0.01$ \\ 
    3 & GO Biological Process & regulation of gene expression & $<0.01$ & $<0.01$ \\ 
    3 & GO Biological Process & regulation of RNA metabolic process & $<0.01$ & $<0.01$ \\ 
    3 & GO Biological Process & regulation of nucleobase-containing compound metabolic process  & $<0.01$ & $<0.01$ \\ 
    3 & GO Biological Process & regulation of RNA biosynthetic process & $<0.01$ & $<0.01$ \\ 
    3 & GO Molecular Function & nucleic acid binding & $<0.01$ & $<0.01$ \\ 
    3 & GO Molecular Function & heterocyclic compound binding  & $<0.01$ & $<0.01$ \\ 
    3 & GO Molecular Function & organic cyclic compound binding  & $<0.01$ & $<0.01$ \\ 
    3 & GO Molecular Function & DNA binding & $<0.01$ & $<0.01$ \\ 
    3 & GO Molecular Function & RNA binding & $<0.01$ & $<0.01$ \\ 
    3 & GO Cellular Component & nucleoplasm  & $<0.01$ & $<0.01$ \\ 
    3 & GO Cellular Component & nucleus  & $<0.01$ & $<0.01$ \\ 
    3 & GO Cellular Component & nuclear lumen  & $<0.01$ & $<0.01$ \\ 
    3 & GO Cellular Component & nuclear protein-containing complex  & $<0.01$ & $<0.01$ \\ 
    3 & GO Cellular Component & intracellular organelle  & $<0.01$ & $<0.01$ \\ 
   \hline
\end{tabular}
}
\end{table}

\newpage

\section{Additional results about the differential correlation analysis of genes}\label{sec:othermethods}

Figure~\ref{fig:Cai-hist} shows the spectral projection scores of \cite{cai2017computational} on the 2000 genes. It can be seen that there is no clear cluster pattern observed (compared with Figure~\ref{fig:SS-hist}). Running Kmeans algorithm of \cite{cai2017computational} would result in a

\begin{figure}[H]
\centering
\includegraphics[width=0.75\textwidth]{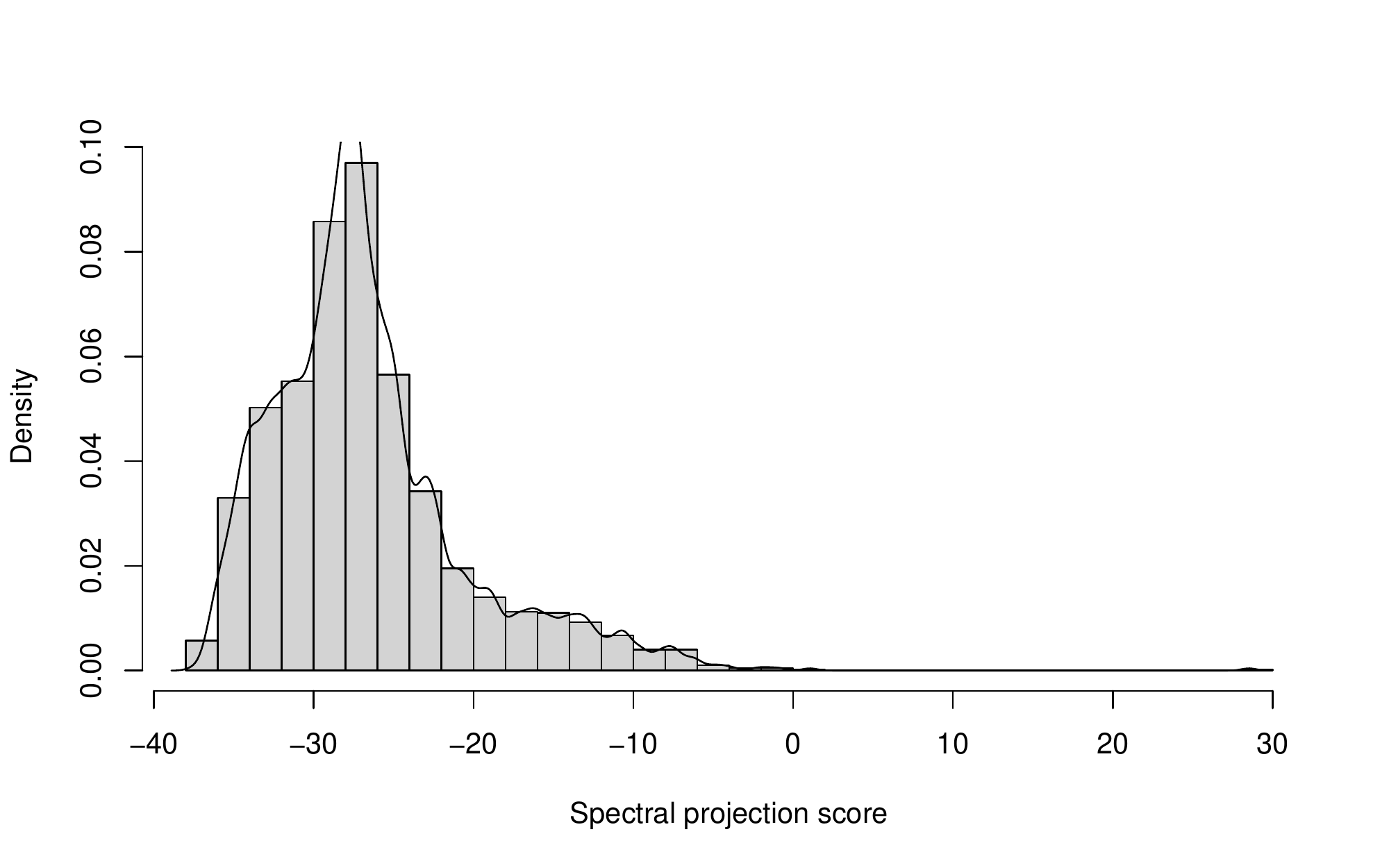}
\caption{The histogram and density of the scores 2000 genes by the spectral projection method of \cite{cai2017computational}.}
\label{fig:Cai-hist}
\end{figure}

The adaptive LAS selects 350 genes. The GSS selects 175 genes, which is a strict subset of the adaptive LAS selection. Next, we will mainly compare our results with the adaptive LAS selection. The average correlation level for the selected genes is 0.54 in the GBM group, and 0.48 in the normal group. Compared with the selection of our method, this gap is much smaller. Furthermore, we repeat the same type of network analysis on the 350 selected genes from the adaptive LAS. The networks and the detected modules are shown in Figure~\ref{fig:HCD-LAS}.  The variation patterns of module 3 and 4 between the GBM and normal group are not consistently reflected in the LGG group. Therefore, the result is considered as inferior to the result based on the spectral screening selection.

\begin{figure}[H]
\centering
\includegraphics[width=\textwidth]{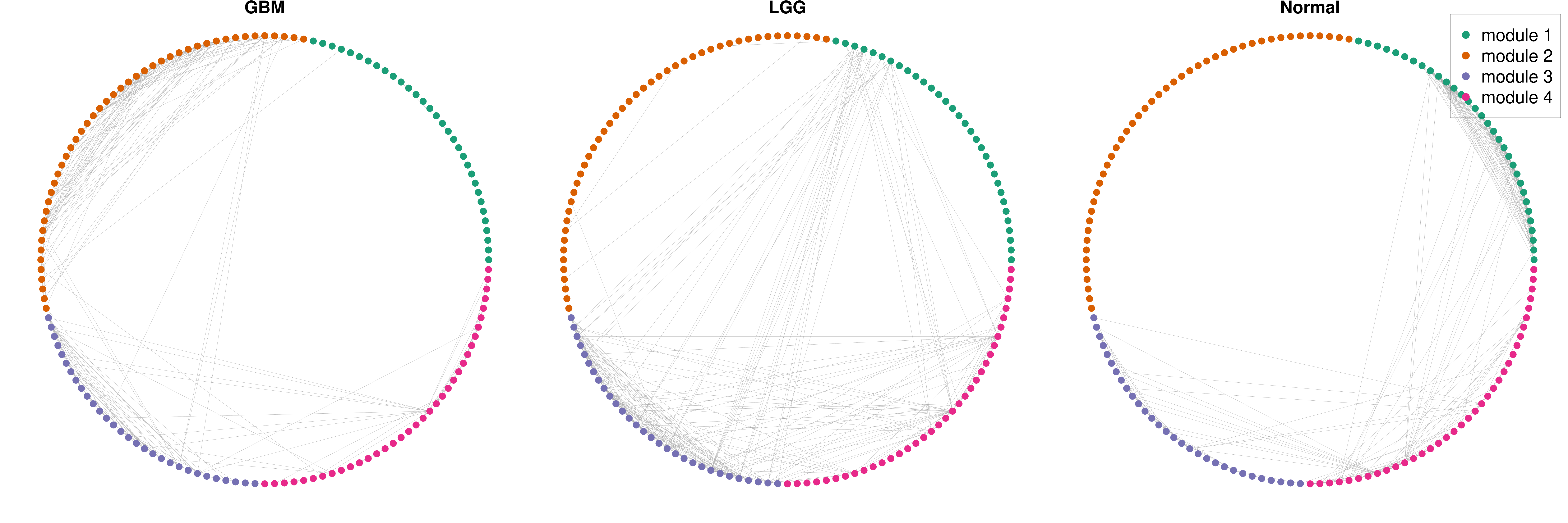}
\caption{The networks of GBM, LGG and normal groups and the 4 gene modules identified by the HCD algorithm based on the 350 genes selected by the adaptive LAS algorithm.}
\label{fig:HCD-LAS}
\end{figure}

\end{appendix}

\end{document}